\documentclass[12pt]{article}
\usepackage{amsmath}
\usepackage{graphicx,psfrag,epsf}
\usepackage{enumerate}
\usepackage{url} 

\usepackage{xr-hyper}
    \usepackage[dvipsnames]{xcolor}
    \usepackage[T1]{fontenc}
	\usepackage[utf8]{inputenc}
	\usepackage[english]{babel}
	\usepackage{amsfonts}
	\usepackage{amssymb}
	\usepackage{amsmath}
	
	\usepackage{enumerate}
	\usepackage{bbm}
	\usepackage{mathtools}
	\usepackage{subfig}
	\usepackage{booktabs}	
	\usepackage{setspace}
	\usepackage{graphicx}
	\usepackage{fullpage}
	
	\usepackage{soul}
	\usepackage{xspace}
	\usepackage{csquotes}
	\usepackage[margin = 1in]{geometry}
	\usepackage{placeins}
	\usepackage[footfont=normalsize,font=normalsize]{floatrow}
	\usepackage{lscape}
	\usepackage{longtable}
	\usepackage{mathabx}
    \usepackage{accents}
    \usepackage{float}

    \usepackage{tikz}
    \usetikzlibrary{calc,intersections}

	\usepackage[longnamesfirst]{natbib}
	\usepackage{bibunits}
	\defaultbibliography{Bibliography.bib}  
	\defaultbibliographystyle{aer}

	\usepackage{hyperref}
	\hypersetup{
    colorlinks=true,
    linkcolor=Red,
    filecolor=magenta,      
    urlcolor=blue,
    citecolor = Blue
}
	\usepackage{cleveref}

    \usepackage{rotating}
	\usepackage[framemethod=tikz]{mdframed}

    \usepackage{tcolorbox}
    \newtcbox{\feedback}{nobeforeafter,colframe=black,colback=white,boxrule=0.5pt,arc=2pt,
      boxsep=0pt,left=2pt,right=2pt,top=2pt,bottom=2pt,tcbox raise base}
      
    \usepackage{amsthm}
    \usepackage{thmtools}
    \renewcommand{\thmcontinues}[1]{continued}

    \newtheorem{asm}{Assumption}
    
    \newtheorem{prop}{Proposition}[section]
    \newtheorem{lem}{Lemma}[section]
    
    \theoremstyle{definition}

    \newtheorem{example}{Example}

\declaretheoremstyle[
  bodyfont=\normalfont,
  postheadspace=1em,
  qed=$\blacktriangle$
]{mythmstyle}

\declaretheorem[style=mythmstyle]{Remark}

\usepackage{array}
\newcolumntype{L}[1]{>{\raggedright\let\newline\\\arraybackslash}m{#1}}
\newcolumntype{C}[1]{>{\centering\let\newline\\\arraybackslash\hspace{0pt}}m{#1}}
\newcolumntype{R}[1]{>{\raggedleft\let\newline\\\arraybackslash\hspace{0pt}}m{#1}}

\usepackage{ragged2e}
\newlength\ubwidth

	\newcommand{\cov}[1]{\mathop{}\!\textnormal{Cov}\bracks{#1}}

	\newcommand{\bracks}[1]{\left[#1\right]}
	\newcommand{\expe}[1]{\mathbb{E}\bracks{#1}}
	
	\newcommand{\expesub}[2]{\mathbb{E}_{#1}\bracks{#2}}
	
	\newcommand{\expefin}[1]{\expesub{f}{#1}}

	\newcommand{\var}[1]{\mathbb{V}\text{ar}\bracks{#1}}
	\newcommand{\varsub}[2]{\mathbb{V}\text{ar}_{#1}\bracks{#2}}
    \newcommand{\varfin}[1]{\varsub{f}{#1}}
    
	\newcommand{\parens}[1]{\left(#1\right)}

    \newcommand{\prob}[1]{\mathbb{P}\parens{#1}}

	\newcommand{\reals}{\mathbb{R}}

	\newcommand{\betahat}{\ensuremath{\hat{\beta}}}
	
	\newcommand{\thetatilde}{\ensuremath{\tilde{\theta}}}

	\newcommand{\sigmahat}{\ensuremath{\hat{\sigma}}}
	
\newcommand{\normnot}[2]{\mathcal{N}\parens{#1,\,#2}}

\newcommand{\twovec}[2]{\parens{\begin{array}{c} #1 \\ #2 \end{array}}}

\newcommand{\twobytwomat}[4]{\parens{\begin{array}{cc} #1 & #2 \\ #3 & #4 \end{array}}}

\newcommand{\gammahat}{\hat{\gamma}}

\newcommand{\tauhat}{\hat\tau}

\newcommand{\thetahat}{\hat{\theta}}

\title{Efficient Estimation for Staggered Rollout Designs\thanks{We are grateful to Bocar Ba, Yuehao Bai, Brantly Callaway, Ivan Canay, Clément de Chaisemartin, Jen Doleac, Peng Ding, Avi Feller, Ryan Hill, Lihua Lei, David McKenzie, Emily Owens, Ashesh Rambachan, Roman Rivera, Evan Rose, Adrienne Sabety, Jesse Shapiro, Yotam Shem-Tov, Dylan Small, Ariella Kahn-Lang Spitzer, Sophie Sun, and seminar participants at Columbia, EU Joint Research Centre, Insper, Notre Dame, PSE/CREST, Seoul National University, UC-Berkeley, University of Cambridge, University of Delaware, University of Florida, University of Mannheim, University of Maryland, University of Pennsylvania, University of Strathclyde, University of Virginia, West Virginia University, the North American Summer Meetings of the Econometric Society, the International Association of Applied Econometrics annual meeting, the Interactions Conference at the University of Wisconsin, and the XVII Escola de Modelos de Regress\~ao for helpful comments and conversations. We thank Madison Perry for excellent research assistance.}}
\author{Jonathan Roth\thanks{Brown University. \href{mailto:jonathanroth@brown.edu}{jonathanroth@brown.edu}} \and Pedro H.C. Sant'Anna\thanks{Emory University.  \href{mailto:pedrosantanna@causal-solutions.com}{pedrosantanna@causal-solutions.com}}}

\begin{document}
\begin{bibunit}

\maketitle

\begin{abstract}
We study estimation of causal effects in staggered rollout designs, i.e. settings where there is staggered treatment adoption and the timing of treatment is as-good-as randomly assigned. We derive the most efficient estimator in a class of estimators that nests several popular generalized difference-in-differences methods. A feasible plug-in version of the efficient estimator is asymptotically unbiased with efficiency (weakly) dominating that of existing approaches. We provide both $t$-based and permutation-test-based methods for inference. In an application to a training program for police officers, confidence intervals for the proposed estimator are as much as eight times shorter than for existing approaches.
\end{abstract}
\pagebreak
\section{Introduction}

Researchers are often interested in the causal effect of a treatment that is first implemented for different units at different times. Staggered rollouts are frequently analyzed using methods that extend the simple two-period difference-in-differences (DiD) estimator to the staggered setting, such as two-way fixed effects (TWFE) regression estimators and recently-proposed alternatives that yield more intuitive causal parameters under treatment effect heterogeneity \citep{callaway_difference--differences_2020, de_chaisemartin_two-way_2020, sun_estimating_2020}. The validity of these estimators depends on a parallel trends assumption.

However, researchers often justify parallel trends by arguing that the timing of the treatment is as good as randomly assigned. In some settings, such as our application to the rollout of a training program for police officers, the timing of the treatment is explicitly randomized.\footnote{When treatment is as good as randomly assigned, other methods (e.g. simple comparisons of means) are available to estimate average treatment effects. DiD based methods have nevertheless been recommended for randomized rollouts to improve efficiency \citep{xiong_optimal_2019} and to transparently aggregate treatment effect heterogeneity \citep{Lindner2021}.} In other settings, treatment timing is not explicitly randomized, but the researcher argues that it is due to idiosyncratic quasi-random factors. For example, \citet{deshpande_who_2019} justify the use of a DiD design comparing areas whose social security office closed at different times by arguing that the ``timing of the closings appears to be effectively random'', as evidenced by the fact that observable characteristics are balanced across units adopting at different times. DiD and related methods have also been used to exploit the quasi-random timing of parental deaths \citep{nekoei_how_2023}, health shocks \citep{fadlon_family_2021}, and stimulus payments \citep{parker_consumer_2013}, among others.


In this paper, we show that if treatment timing is as good as randomly assigned, one can obtain more precise estimates than those provided by DiD-based methods. We derive the most efficient estimator in a large class of estimators that nests many existing DiD-based approaches, and show how to conduct both $t$-based and permutation-based inference. In settings where treatment timing is as good as random, our efficient estimator has the scope to substantially reduce standard errors, as illustrated in our simulations and application below.

We begin by introducing a design-based framework that formalizes the notion that treatment timing is (quasi-)randomly assigned. There are $T$ periods, and unit $i$ is first treated in period $G_i \in \mathcal{G} \subseteq \{1,...,T, \infty \}$, with $G_i = \infty$ denoting that $i$ is never treated (or treated after period $T$). We make two key assumptions in this model. First, we assume that the treatment timing $G_i$ is (quasi-)randomly assigned, in the sense that any permutation of the observed vector of treatment start dates is equally likely to occur. Second, we rule out anticipatory effects of treatment --- for example, a unit's outcome in period two does not depend on whether it was first treated in period three or in period four. 

Within this framework, we show that pre-treatment outcomes play a similar role to fixed covariates in a randomized experiment, and generalized DiD estimators can be viewed as applying a crude form of covariate adjustment. To develop intuition, it is instructive to first consider the special case where we observe data for two periods $(T=2)$, some units are first treated in period 2 $(G_i =2)$, and the remaining units are treated in a later period or never treated $(G_i = \infty)$. This special case is analogous to conducting a randomized experiment in period 2, with the outcome in period 1 serving as a pre-treatment covariate. The DiD estimator is $\thetahat^{DiD} = (\bar{Y}_{22} - \bar{Y}_{2\infty}) - (\bar{Y}_{12} - \bar{Y}_{1\infty})$, where $\bar{Y}_{tg}$ is the mean outcome for treatment group $g$ at period $t$. It is clear that $\thetahat^{DiD}$ is a special case of the class of estimators
\begin{equation}
\thetahat_\beta = \underbrace{(\bar{Y}_{22} - \bar{Y}_{2\infty})}_{\text{Post-treatment diff}} - \beta \underbrace{(\bar{Y}_{12} - \bar{Y}_{1\infty})}_{\text{Pre-treatment diff}} \label{eqn: thetahat beta for twoperiod}
\end{equation}
\noindent which adjust the post-treatment difference in means by $\beta$ times the pre-treatment difference in means. Under the assumption of (quasi-)random treatment timing, $\thetahat_\beta$ is unbiased for the average treatment effect (ATE) for any $\beta$, since the post-treatment difference in means is unbiased for the ATE and the pre-treatment difference in means is mean-zero. The value of $\beta$ that minimizes the variance of the estimator depends on the covariances of the potential outcomes between periods, however. Intuitively, we want to put more weight on lagged outcomes when they are more informative about post-treatment outcomes. DiD, which imposes the fixed weight $\beta=1$, will thus generally be inefficient, and one can obtain an (asymptotically) more efficient estimator by estimating the optimal weights from the data. In this special two-period case, the form of the efficient estimator follows from \citet{lin_agnostic_2013}, who studied efficient covariate adjustment in cross-sectional randomized experiments; see, also, \citet{mckenzie_beyond_2012} who noted that the two-period DiD estimator may be inefficient in experiments.

Our main theoretical results extend this logic to the case of staggered treatment timing, providing formal methods for estimation and inference. We begin by introducing a flexible class of causal parameters that can highlight treatment effect heterogeneity across both calendar time and time since treatment. Following \citet{athey_design-based_2022}, we define $\tau_{t,gg'}$ to be the average effect on the outcome in period $t$ of changing the initial treatment date from $g'$ to $g$. For example, in the simple two-period case described above, $\tau_{2,2\infty}$ corresponds with the average treatment effect (ATE) on the second-period outcome of being treated in period two relative to never being treated. We then consider the class of estimands that are linear combinations of these building blocks, $\theta = \sum_{t,g,g'} a_{t,g,g'} \tau_{t,gg'}$. Our framework thus allows for arbitrary treatment effect dynamics, and accommodates a variety of ways of summarizing these dynamic effects, including several aggregation schemes proposed in the recent literature.

We then consider the large class of estimators that start with a sample analog to the target parameter and adjust by a linear combination of differences in pre-treatment outcomes. More precisely, we consider estimators of the form $\thetahat_\beta = \sum_{t,g} a_{t,g,g'} \tauhat_{t,gg'} - \hat{X}' \beta$, where the first term is a sample analog to $\theta$, and the second term adjusts linearly using a vector $\hat{X}$ that compares outcomes for cohorts treated at different dates at points in time before either was treated. For example, in the simple two-period case described above, $\hat{X} = \bar{Y}_{12} - \bar{Y}_{1\infty}$ is the difference-in-means in period 1. We show that several estimators for the staggered setting are part of this class for an appropriately defined estimand and $\hat{X}$, including the TWFE estimator as well as recent procedures proposed by \citet{callaway_difference--differences_2020}, \citet{de_chaisemartin_two-way_2020}, and \citet{sun_estimating_2020}. All estimators of this form are unbiased for $\theta$ under the assumptions of (quasi-)random treatment timing and no anticipation. 

We then derive the most efficient estimator in this class. The optimal coefficient $\beta^*$ depends on covariances between the potential outcomes over time, and thus the estimators previously proposed in the literature will only be efficient for special covariance structures. Although the covariances of the potential outcomes are generally not known ex ante, one can estimate a ``plug-in'' version of the efficient estimator that replaces the ``oracle'' coefficient $\beta^*$ with a sample analog $\betahat^*$. We show that the plug-in efficient estimator is asymptotically unbiased and as efficient as the oracle estimator under large population asymptotics similar to those in \citet{lin_agnostic_2013} and \citet{li_general_2017} for covariate adjustment in cross-sectional experiments. 

Our results suggest two complementary approaches to inference. First, we show that the plug-in efficient estimator is asymptotically normally distributed in large populations, which allows for asymptotically valid confidence intervals of the familiar form $\thetahat_{\betahat^*} \pm 1.96 \widehat{se}$.\footnote{As is common in finite-population settings, the covariance estimate may be conservative if there are heterogeneous treatment effects.} Second, an appealing feature of our (quasi-)random treatment timing framework is that it permits us to construct Fisher randomization tests (FRTs), also known as permutation tests. Following  \citet{wu_randomization_2020} and \citet{zhao_covariate-adjusted_2020} for cross-sectional randomized experiments, we consider FRTs based on a studentized version of our efficient estimator. These FRTs have the dual advantages that they are finite-sample exact under the sharp null of no treatment effects, and asymptotically valid for the weak null of no average effects. In a Monte Carlo study calibrated to our application, we find that both the $t$-based and FRT-based approaches yield reliable inference, and CIs based on the plug-in efficient estimator are substantially shorter than those for the procedures of \citet{callaway_difference--differences_2020}, \citet{sun_estimating_2020}, and \citet{de_chaisemartin_two-way_2020}.\footnote{The \texttt{staggered} R and Stata packages allow for easy implementation of the plug-in efficient estimator; see \href{https://github.com/jonathandroth/staggered}{https://github.com/jonathandroth/staggered} and \href{https://github.com/mcaceresb/stata-staggered}{https://github.com/mcaceresb/stata-staggered}, respectively.}

As an illustration of our method and standalone empirical contribution, we revisit the randomized rollout of a procedural justice training program for police officers in Chicago. The original study by \citet[]{wood_procedural_2020} found large and statistically significant reductions in complaints and officer use of force, and these findings were influential in policy debates about policing \citep{doleac_how_2020}. Unfortunately, an earlier version of our analysis revealed a statistical error in the analysis of \citet{wood_procedural_2020} which led their estimates to be inflated. In \citet{wood_reanalysis_2020}, we collaborated with the original authors to correct this error. Using the estimator of \citet{callaway_difference--differences_2020}, we found no significant effects on complaints against police officers and borderline significant effects on officer use of force, but with wide confidence intervals that included both near-zero and meaningfully large treatment effects estimates. We find that the use of the methodology proposed in this paper allows us to obtain substantially more precise estimates of the effect of the training program. Although we again find no statistically significant effects on complaints and borderline significant effects on force, the standard errors from using our methodology are between 1.4 and 8.4 times smaller than from the \citet{callaway_difference--differences_2020} estimator used in \citet{wood_reanalysis_2020}. For complaints, for example, we are able to rule out reductions larger than 13\% of the pre-treatment mean using our proposed estimator, compared with an upper bound of 33\% in the previous analysis.

\paragraph{Related Literature.} This paper contributes to an active literature on DiD and related methods in settings with staggered treatment timing. Several recent papers have demonstrated the failures of TWFE models to recover a sensible causal estimand under treatment effect heterogeneity and have proposed alternative estimators with better properties \citep{borusyak_revisiting_2016,goodman-bacon_difference--differences_2018, de_chaisemartin_two-way_2020,callaway_difference--differences_2020,sun_estimating_2020}. Most of this literature has focused on obtaining consistent estimates under a generalized parallel trends assumption, whereas we focus on \textit{efficient} estimation under the stronger assumption of (quasi-)random treatment timing. Our proposed efficient estimator can help to improve precision relative to DiD methods in settings where the researcher believes that treatment timing is as good as randomly assigned, but unlike other estimators in the literature, will not be applicable in settings where the researcher is confident in parallel trends but not (quasi-)random treatment timing. For example, our random treatment timing assumption requires that pre-treatment outcomes and fixed covariates should be balanced across groups treated at different times, whereas this is not strictly required by parallel trends. See Remark \ref{rem: pt vs random timing} for further discussion.

Two related papers that have studied (quasi-)random treatment timing are \citet{athey_design-based_2022} and \citet{shaikh_randomization_2021}. The former studies a model of random treatment timing similar to ours, but focuses on the interpretation of the TWFE estimand. The latter paper adopts a different framework of randomization in which treatment timing is random only conditional on observables, and no two units can be treated at the same time. Neither paper considers the efficient choice of estimator as we do.

Our technical results extend results in statistics on efficient covariate adjustment in cross-sectional experiments \citep{ Freedman(2008)-several_treatments, Freedman(2008)-regadj_to_experimental_data, lin_agnostic_2013, li_general_2017} to the setting of staggered treatment timing, where pre-treatment outcomes play a similar role to fixed covariates in a cross-sectional experiment. In the special two-period case, our proposed estimator reduces to \citet{lin_agnostic_2013}'s efficient estimator, treating the lagged outcome as a fixed covariate. Our results are also related to \citet{mckenzie_beyond_2012}, who showed that DiD is inefficient under random treatment assignment in a two-period model with homogeneous treatment effects; see Remark \ref{rem: connection to mckenzie} for additional details. We note that the notion of efficiency studied in this paper is efficiency in the class of estimators of the form given in (\ref{eqn: thetahat beta for twoperiod}), rather than semi-parametric efficiency as in e.g. \citet{hahn_role_1998}, \citet{santanna_doubly_2020}. We are not aware of a notion of semi-parametric efficiency for design-based models such as ours, but consider this an interesting topic for future work. 

Our paper also relates to the literature on clinical trials using a stepped wedge design, which is a randomized staggered rollout in which all units are ultimately treated \citep[e.g.][]{Brown2006b}. Until recently, this literature has focused on estimation using mixed effects regression models. \citet{Lindner2021} point out, however, that such models may be difficult to interpret under treatment effect heterogeneity, and recommend using DiD-based approaches like \citet{sun_estimating_2020} instead. Our approach has the potential to offer large gains in precision relative to such DiD-based approaches. Our paper is also complementary to \citet{Ji2017}, who propose using randomization-based inference procedures to test Fisher's sharp null hypothesis in stepped wedge designs. By contrast, we consider Neymanian inference on average treatment effects, and also show that an FRT with a studentized test statistic is both finite-sample exact for the sharp null and asymptotically valid for inference on average effects. 


Finally, our work is related to \citet{xiong_optimal_2019} and \citet{basse_minimax_2023}, who consider the optimal design of a staggered rollout experiment to maximize the efficiency of a fixed estimator. By contrast, we solve for the most efficient estimator given a fixed experimental design.


\section{Model and Theoretical Results \label{sec: multiple periods}}

\subsection{Model}

There is a finite population of $N$ units. We observe data for $T$ periods, $t = 1,..,T$. A unit's treatment status is denoted by $G_i \in \mathcal{G} \subseteq \{1,...,T, \infty \}$, where $G_i$ is the first period in which unit $i$ is treated, and $G_i = \infty$ denotes that a unit is never treated (or treated after period $T$). Our framework accommodates but does not require there to be never treated units --- it could be that $\infty \not \in \mathcal{G}$, in which case all units are eventually treated (a stepped wedge design). We assume that treatment is an absorbing state.\footnote{If treatment turns on and off, the parameters we estimate can be viewed as the intent-to-treat effect of first being treated at a particular date; see \citet{sun_estimating_2020} and \citet{de_chaisemartin_difference--differences_2021} for related discussion for DiD models.} We denote by $Y_{it}(g)$ the potential outcome for unit $i$ in period $t$ when treatment starts at time $g$, and define the vector $Y_i(g) = (Y_{i1}(g),...,Y_{iT}(g))' \in \reals^T$. We let $D_{ig} = 1[G_i = g]$. The observed vector of outcomes for unit $i$ is then $Y_i = \sum_g  D_{ig} Y_i(g)$.

Following \citet{neyman_application_1923} for randomized experiments and \citet{athey_design-based_2022} for settings with staggered treatment timing, our model is design-based: We treat as fixed (or condition on) the potential outcomes and the number of units first treated at each period $(N_g)$. The only source of uncertainty in our model comes from the vector of times at which units are first-treated, $G = (G_1,...,G_N)'$, which is stochastic.

\begin{Remark}[Design-based uncertainty]
Design-based models are particularly attractive in settings where it is difficult to define the super-population, such as when all 50 states are observed \citep{ManskiPepper(18)}, or in our application where the near-universe of police officers in Chicago is observed. Even when there is a super-population, the design-based view allows for valid inference on the sample average treatment effect (SATE); see \citet{abadie_sampling-based_2020}, \citet{sekhon_inference_2020} for additional discussion.
\end{Remark}

Our first main assumption is that the treatment timing is (quasi-)randomly assigned, meaning that any permutation of the treatment timing vector is equally likely. 

\begin{asm}[Random treatment timing] \label{asm: random treatment}
Let $D$ be the random $N \times |\mathcal{G}|$ matrix with $(i,g)$th element $D_{ig}$. Then $\prob{D = d} = (\prod_{g \in \mathcal{G}} N_g! ) / N!$ if $\sum_i  d_{ig} = N_g$ for all $g$, and zero otherwise. 
\end{asm}
\noindent We note that Assumption \ref{asm: random treatment} will hold by design in settings where the researcher randomly assigns individuals to treatment start dates. It can also hold in quasi-experimental contexts if the idiosyncratic factors that determine treatment timing render any permutation of the treatment start dates to be equally likely; see \citet{rambachan_design-based_2020} and \citet{borusyak_non-random_2020} for additional discussion of ``quasi-random'' treatment assignment. We discuss extensions to clustered and conditional 
random assignment of treatment timing in Section \ref{subsec: extensions}.

\begin{Remark}[Comparison to parallel trends] \label{rem: pt vs random timing}
Technically speaking, the random timing assumption in Assumption \ref{asm: random treatment} is stronger than the usual parallel trends assumption, which only requires that treatment probabilities are orthogonal to trends in the potential outcomes. Assumption \ref{asm: random treatment} thus may not be plausible in all settings where researchers use DiD methods. Nevertheless, Assumption \ref{asm: random treatment} can be ensured by design in settings where treatment timing can be explicitly randomized, such as our application in Section \ref{sec: wood application}. Moreover, it is frequently the case that the justification given for the validity of the parallel-trends assumption also justifies Assumption \ref{asm: random treatment}.\footnote{Analogously, \citet[page 8]{Imbens2004} argues that while mean-independence is technically weaker than full independence, arguments for the former often also justify the latter.}  For example, \citet{fadlon_family_2021} write that the plausibility of the parallel trends assumption in their context ``relies on the notion that... the particular year at which the event occurs may be as good as random'' (p. 12-13); see, e.g., \citet{deshpande_who_2019}, \citet{nekoei_how_2023}, and \citet{parker_consumer_2013} for similar justifications.

It is also worth emphasizing that in non-experimental contexts, the random timing assumption may be more plausible if one restricts attention to units who are eventually treated. For example, \citet[page 223]{deshpande_who_2019} write that ``some factors consistently predict the likelihood of a closing [i.e., the treatment]. However, no observable characteristic consistently predicts the timing of a closing conditional on closing. These results suggest that the timing of closings is effectively random even if the closings themselves are not.'' Although in principle one can use DiD methods to exploit variation only among eventually-treated units, units who are never-treated are often included in DiD analyses to increase precision.\footnote{For example, the main specification in \citet{bailey_war_2015} includes never-treated units, although the appendix shows results for an alternative specification that includes only eventually-treated units, with substantially larger standard errors (contrast Figures 5 and E.1).} In settings where the eventually-treated units are more similar to each other than to the never-treated units, it therefore may be preferable to impose Assumption \ref{asm: random treatment} and use our efficient estimator than to use a DiD estimator that relies on parallel trends among never-treated units to increase efficiency. We also note that Assumption \ref{asm: random treatment} has testable implications, as we discuss in Section \ref{subsec: extensions} below, so researchers considering using our methodology in non-experimental contexts can partially test the validity of Assumption \ref{asm: random treatment}.

Finally, we note that the validity of the parallel trends assumption will typically be sensitive to functional form if treatment timing is not random \citep{roth_when_2022}. Empirical researchers should therefore be explicit about the justification for identification. If parallel trends is justified on the basis of quasi-random treatment timing, then the methods developed in this paper can be used to obtain more precise estimates. On the other hand, if random treatment timing is not plausible, then methods that rely only on a parallel trends assumption will be more appropriate. In this case, however, the researcher should provide a justification for why they expect parallel trends to hold specifically for the choice of functional form used in the analysis. 
\end{Remark}

In addition to random treatment timing, we also assume that the treatment has no causal impact on the outcome in periods before it is implemented. This assumption is plausible in many contexts, but may be violated if individuals learn of treatment status beforehand and adjust their behavior in anticipation \citep{abbring_nonparametric_2003,lechner_estimation_2010,malani_interpreting_2015}.\footnote{If anticipatory behavior is only possible within $m$ periods of treatment (e.g., because treatment is announced $m$ periods in advance), the initial treatment can be re-defined as $G_i - m$.}

\begin{asm}[No anticipation] \label{asm: no anticipation}
For all $i$, $Y_{it}(g) = Y_{it}(g')$ for all $g,g' > t$.
\end{asm}

\noindent Note that this assumption does not restrict the possible dynamic effects of treatment --- that is, we allow for $Y_{it}(g) \neq Y_{it}(g')$ whenever $t \geq min(g,g')$, so that treatment effects can arbitrarily depend on calendar time and the time that has elapsed since treatment. Rather, we only require that, say, a unit's outcome in period one does not depend on whether it was ultimately treated in period two or period three.\footnote{Under the No Anticipation Assumption, $Y_{it}(g)$ can be interpreted as the outcome in period $t$ from having been treated for $\max(0,t-g)$ periods. We thank a referee for noting this interpretation.\label{fn: max interp of POs}}

\begin{example}[Special case: two periods]
\label{example:2periods}
Consider the special case of our model in which there are two periods $(T=2)$ and units are either treated in period two or never treated $(\mathcal{G} = \{2,\infty\})$. Under random treatment timing and no anticipation, this special case is isomorphic to a cross-sectional experiment where the outcome $Y_{i} = Y_{i2}$ is the second period outcome, the binary treatment $D_i = 1[G_i = 2]$ is whether a unit is treated in period two, and the covariate $X_i = Y_{i1} \equiv Y_{i1}(\infty)$ is the pre-treatment outcome (which by the no anticipation assumption does not depend on treatment status). Covariate adjustment in cross-sectional randomized experiments has been studied previously by \citet{Freedman(2008)-several_treatments, Freedman(2008)-regadj_to_experimental_data}, \citet{lin_agnostic_2013}, and \citet{li_general_2017}, and our results will nest many of the existing results in the literature as a special case. The two-period special case also allows us to study the canonical difference-in-differences estimator, while avoiding complications discussed in the recent literature related to extending this estimator to the staggered case. We will therefore come back to this example throughout the paper to provide intuition and connect our results to the previous literature. 

\end{example}

\paragraph{Notation.} All expectations $(\expe{\cdot})$ and probability statements $(\prob{\cdot})$ are taken over the distribution of $G$ conditional on the potential outcomes and the number of units treated at each period, $(N_g)_{g\in \mathcal{G}}$, although we suppress this conditioning for ease of notation. For a non-stochastic attribute $W_i$ (e.g. a function of the potential outcomes), we denote by $\expefin{W_i} = {N}^{-1} \sum_i  W_i$ and $\varfin{W_i} = {(N-1)}^{-1} \sum_i  (W_i - \expefin{W_i})(W_i-\expefin{W_i})'$ the finite-population expectation and variance of $W_i$.

\subsection{Target Parameters\label{subsec: target parameter}}

In our staggered treatment setting, the effect of being treated may depend on both the calendar time ($t$) as well as the time at which one was first treated ($g$). We therefore consider a large class of target parameters that allow researchers to highlight various dimensions of heterogeneous treatment effects across both calendar time and time since treatment. 

Following \citet{athey_design-based_2022}, we define $\tau_{it,gg'} = Y_{it}(g) - Y_{it}(g')$ to be the causal effect of switching the treatment date from $g'$ to $g$ on unit $i$'s outcome in period $t$. We define $\tau_{t,gg'} = {N}^{-1} \sum_i  \tau_{it,gg'}$ to be the average treatment effect (ATE) of switching treatment from $g'$ to $g$ on outcomes at period $t$. We will consider scalar estimands of the form \begin{equation}\theta = \sum_{t,g,g'} a_{t,gg'} \tau_{t,gg'} \label{eqn: defn of theta in terms of tau},\end{equation} i.e. weighted sums of the average treatment effects of switching from treatment $g'$ to $g$, with $a_{t,gg'}\in \reals$ being arbitrary weights. Researchers will often be interested in weighted averages of the $\tau_{t,gg'}$, in which case the $a_{t,gg'}$ will sum to 1, although our results allow for arbitrary $a_{t,gg'}$.\footnote{This allows the possibility, for instance, that $\theta$ represents the difference between long-run and short-run effects, so that some of the $a_{t,gg'}$ are negative.} The results extend easily to vector-valued $\theta$'s where each component is of the form in the previous display; we focus on the scalar case for ease of notation. The no anticipation assumption (Assumption \ref{asm: no anticipation}) implies that $\tau_{t,gg'}=0$ if $t<min(g,g')$, and so without loss of generality we make the normalization that $a_{t,gg'} =0$ if $t<min(g,g')$.

\begin{example}[continues=example:2periods]
In our simple two-period example, a natural target parameter is the ATE in period two. This corresponds with setting $\theta = \tau_{2,2\infty} = N^{-1} \sum_i Y_{i2}(2) - Y_{i2}(\infty)$.
\end{example}

We now describe a variety of intuitive parameters that can be captured by this framework in the general staggered setting. Researchers are often interested in the effect of receiving treatment at a particular time relative to not receiving treatment at all. We will define $ATE(t,g) := \tau_{t,g\infty}$ to be the average treatment effect on the outcome in period $t$ of being first-treated at period $g$ relative to not being treated at all. The $ATE(t,g)$ is a close analog to the cohort average treatment effects on the treated (ATTs) considered in \citet{callaway_difference--differences_2020} and \citet{sun_estimating_2020}. The main difference is that those papers do not assume random treatment timing, and thus consider ATTs rather than ATEs.

In some cases, the $ATE(t,g)$ will be directly of interest and can be estimated in our framework. When the dimension of $t$ and $g$ is large, however, it may be desirable to aggregate the $ATE(t,g)$ both for ease of interpretability and to increase precision. Our framework incorporates a variety of possible summary measures that aggregate the $ATE(t,g)$ across different cohorts and time periods. We briefly discuss a few possible aggregations which may be relevant in empirical work, mirroring proposals for aggregating the $ATT(t,g)$ in \citet{callaway_difference--differences_2020}. 

When researchers are interested in how the treatment effect evolves with respect to the time elapsed since treatment started, they may want to consider ``event-study'' parameters that aggregate the ATEs at a given lag $l$ since treatment ($l=0,1,...$), 
$$\theta^{ES}_l = \frac{1}{\sum_{g:g+l \leq T} N_g} \sum_{g:g+l \leq T} N_g ATE(g+l,g).$$
\noindent Note that the instantaneous parameter $\theta^{ES}_0$ is analogous to the estimand considered in \citet{de_chaisemartin_two-way_2020} in settings like ours where treatment is an absorbing state (although their framework also extends to the more general setting where treatment turns on and off).

In other situations, it may be of interest to understand how the treatment effect differs over calendar time (e.g. during a boom or bust economy), or by the time that treatment began. In such cases, the summary parameters 
$$\theta_t = \frac{1}{ \sum_{g:g \leq t} N_g } \sum_{g: g \leq t} N_g ATE(t,g) \text{ and } \theta_g = \frac{1}{ T-g+1 } \sum_{t: t\geq g } ATE(t,g),$$
\noindent which respectively aggregate the $ATE$s for a particular calendar time or treatment adoption cohort, may be relevant. 

Finally, researchers may be interested in a single summary parameter for the effect of a treatment. In this case, it may be instructive to consider a simple average of the $ATE(t,g)$ (weighted by cohort size), 
$$\theta^{simple} = \frac{1}{ \sum_t \sum_{g:g \leq t} N_g } \sum_{t} \sum_{g:g \leq t} N_g ATE(t,g),$$
or to consider a weighted average of the time or cohort effects,
$$\theta^{calendar} = \frac{1}{T} \sum_t \theta_t \hspace{1cm} \text{ or } \hspace{1cm} \theta^{cohort} =\frac{1}{\sum_{g:g\neq \infty}N_g} \sum_{g:g\neq \infty} N_g \theta_g.$$

Since the most appropriate parameter will depend on context, we consider a broad framework that allows for efficient estimation of all of these (and other) parameters.\footnote{We note that if $\infty \not\in \mathcal{G}$, then $ATE(t,g)$ is only identified for $t<\max{\mathcal{G}}.$ In this case, all of the sums above should be taken only over the $(t,g)$ pairs for which $ATE(t,g)$ is identified.}

\subsection{Class of Estimators Considered\label{subsec: estimators considered}}

We now introduce the class of estimators we will consider. Intuitively, these estimators start with a sample analog to the target parameter and linearly adjust for differences in outcomes for units treated at different times in periods before either was treated.

Let $\bar{Y}_{tg} = {N_g}^{-1} \sum_i  D_{ig} Y_{it}$ be the sample mean of the outcome for treatment group $g$ in period $t$, and let $\tauhat_{t,gg'} = \bar{Y}_{tg} - \bar{Y}_{tg'}$ be the sample analog of $\tau_{t,gg'}$. We define 
$$\thetahat_{0} = \sum_{t,g,g'} a_{t,gg'} \tauhat_{t,gg'}, $$
\noindent which replaces the population means in the definition of $\theta$ with their sample analogues.

We will consider estimators of the form 
\begin{equation}
\thetahat_\beta = \thetahat_0 - \hat{X}' \beta, \label{eqn: defn of thetahatbeta} 
\end{equation}
where, intuitively, $\hat{X}$ is a vector of differences-in-means that are guaranteed to be mean-zero under the assumptions of random treatment timing and no anticipation. Formally, we consider $M$-dimensional vectors $\hat X$ where each element of $\hat{X}$ takes the form $$\hat X_j = \sum_{(t,g,g'):g,g'>t} b^j_{t,gg'} \tauhat_{t,gg'},$$ 
where the $b^j_{t,gg'}\in \reals$ are arbitrary weights. There are many possible choices for the vector $\hat X$ that satisfy these assumptions. For example $\hat{X}$ could be a vector where each component equals $\tauhat_{t,gg'}$ for a different combination of $(t,g,g')$ with $t<g,g'$. Alternatively, $\hat{X}$ could be a scalar that takes a weighted average of such differences. The choice of $\hat{X}$ is analogous to the choice of which variables to control for in a cross-sectional randomized experiment. In principle, including more covariates (higher-dimensional $\hat{X}$) will improve asymptotic precision, yet including ``too many'' covariates may lead to over-fitting, leading to poor performance in practice. For now, we suppose the researcher has chosen a fixed $\hat{X}$ and consider the optimal choice of $\beta$ for a given $\hat{X}$. We return to the choice of $\hat X$ in Remark \ref{rem: choice of xhat - theory} and the discussion of our Monte Carlo results in Section \ref{sec: monte carlo} below.

Several estimators proposed in the literature can be viewed as special cases of the class of estimators we consider for an appropriately-defined estimand and $\hat{X}$, often with $\beta=1$.

\begin{example}[continues=example:2periods]
In our running two-period example, $\hat{X} = \tauhat_{1,2\infty}$ corresponds with the pre-treatment difference in sample means between the units first treated at period two and the never-treated units. Thus, $$\thetahat_1 = \tauhat_{2,2\infty} - \tauhat_{1,2\infty} = (\bar{Y}_{22} - \bar{Y}_{2\infty}) - (\bar{Y}_{12} - \bar{Y}_{1\infty}) $$ is the canonical difference-in-differences estimator, where $\bar{Y}_{tg}$ represents the sample mean of $Y_{it}$ for units with $G_i =g$. Likewise, $\thetahat_0$ is the simple difference-in-means (DiM) in period two, $(\bar{Y}_{22} - \bar{Y}_{2\infty})$. More generally, the estimator $\thetahat_\beta$ takes the simple difference-in-means in period two and adjusts by $\beta$ times the difference-in-means in period one. Thus, for $\beta \in (0,1)$, $\thetahat_\beta$ is a weighted average of the DiM and DiD estimators. In this special case, the set of estimators of the form $\thetahat_\beta$ is equivalent to the set of linear covariate-adjusted estimators for cross-sectional experiments considered in \citet{lin_agnostic_2013} and \citet{li_general_2017}, treating $Y_{i1}$ as a fixed covariate.\footnote{\citet{lin_agnostic_2013} and \citet{li_general_2017} consider estimators of the form $\tau(\beta_0, \beta_1) = (\bar{Y}_1 - \beta_1' (\bar{X}_1 - \bar{X}) ) - (\bar{Y}_0 - \beta_0'(\bar{X}_0 - \bar{X}))$, where $\bar{Y}_d$ is the sample mean of the outcome $Y_i$ for units with treatment $D_i = d$, $\bar{X}_d$ is defined analogously, and $\bar{X}$ is the unconditional mean of $X_i$. Setting $Y_i = Y_{i,2}$, $X_i = Y_{i,1}$, and $D_i = 1[G_i = 2]$, it is straightforward to show that the estimator $\tau(\beta_0,\beta_1)$ is equivalent to $\thetahat_{\beta}$ for $\beta = \frac{N_2}{N} \beta_0 + \frac{N_\infty}{N} \beta_1$.}
\end{example}

\begin{example}[\citet{callaway_difference--differences_2020}]\label{example:callway and santa'anna}
For settings where there is a never-treated group ($\infty \in \mathcal{G}$), \citet{callaway_difference--differences_2020} consider the estimator
$$\tauhat^{CS}_{tg} = \tauhat_{t,g \infty} - \tauhat_{g-1,g \infty},$$
\noindent i.e. a difference-in-differences that compares outcomes between periods $t$ and $g-1$ for the cohort first treated in period $g$ relative to the never-treated cohort. Observe that $\tauhat^{CS}_{tg}$ can be viewed as an estimator of $ATE(t,g)$ of the form given in (\ref{eqn: defn of thetahatbeta}), with $\hat{X} = \tauhat_{g-1,g \infty}$ and $\beta =1$. Likewise, \citet{callaway_difference--differences_2020} consider an estimator that aggregates the $\tauhat^{CS}_{tg}$, say $\tauhat_w^{CS} = \sum_{t,g} w_{t,g} \tauhat_{t,g\infty} $, which can be viewed as an estimator of the parameter $\theta_w = \sum_{t,g} w_{t,g} ATE(t,g)$ of the form (\ref{eqn: defn of thetahatbeta}) with $\hat{X} = \sum_{t,g} w_{t,g} \tauhat_{g-1,g\infty}$ and $\beta =1$.\footnote{This could also be viewed as an estimator of the form (\ref{eqn: defn of thetahatbeta}) if $\hat{X}$ were a vector with each element corresponding with $\tauhat_{t,g\infty}$ and the vector $\beta$ was a vector with elements corresponding with $w_{t,g\infty}$.} Similarly, \citet{callaway_difference--differences_2020} consider an estimator that replaces the never-treated group with an average over cohorts not yet treated in period $t$, $$\tauhat^{CS2}_{tg} = \frac{1}{ \sum_{g': g'>t} N_{g'} } \sum_{g': g'>t} N_{g'}  \,  \tauhat_{t,gg'} - \frac{1}{ \sum_{g': g'>t} N_{g'} } \sum_{g': g'>t} N_{g'} \, \tauhat_{g-1,gg'}, \text{ for } t\ge g.$$ It is again apparent that this estimator can be written as an estimator of $ATE(t,g)$ of the form in (\ref{eqn: defn of thetahatbeta}), with $\hat{X}$ now corresponding with a weighted average of $\tauhat_{g-1,gg'}$ and $\beta$ again equal to 1. 
\end{example}    

\begin{example}[\citet{sun_estimating_2020}]\label{example: sun and abraham}
\citet{sun_estimating_2020} consider an estimator that is equivalent to that in \citet{callaway_difference--differences_2020} in the case where there is a never-treated cohort. When there is no never-treated group, \citet{sun_estimating_2020} propose using the last cohort to be treated as the comparison. Formally, they consider the estimator of $ATE(t,g)$ of the form
$$\tauhat^{SA}_{tg} = \tauhat_{t,g g_{max}} - \tauhat_{g-1, g g_{max}},$$
\noindent where $g_{max} = \max \mathcal{G}$ is the last period in which units receive treatment. It is clear that $\tauhat^{SA}_{tg}$ takes the form (\ref{eqn: defn of thetahatbeta}), with $\hat{X} = \tauhat_{g-1, g g_{max}}$ and $\beta = 1$. Weighted averages of the $\tauhat^{SA}_{tg}$ can likewise be expressed in the form (\ref{eqn: defn of thetahatbeta}), as with the \citet{callaway_difference--differences_2020} estimators.
\end{example}

\begin{example}[\citet{de_chaisemartin_two-way_2020}]\label{example: dechaisemartin} 
\citet{de_chaisemartin_two-way_2020} propose an estimator of the instantaneous effect of a treatment. Although their estimator extends to settings where treatment turns on and off, in a setting like ours where treatment is an absorbing state, their estimator can be written as a linear combination of the $\tauhat^{CS2}_{tg}$. In particular, their estimator is a weighted average of the \citet{callaway_difference--differences_2020} estimators for the first period in which a unit was treated,
$$\tauhat^{dCDH} = \frac{1}{\sum_{g:g\leq T} N_g } \sum_{g:g\leq T} N_g \tauhat^{CS2}_{gg}.$$ It is thus immediate from the previous examples that their estimator can also be written in the form (\ref{eqn: defn of thetahatbeta}). 
\end{example}

\begin{example}[TWFE Models]\label{example: twfe}
\citet{athey_design-based_2022} consider the setting with $\mathcal{G} = \{1,...,T,\infty\}$. Let $A_{it} = 1[G_i \leq t ]$ be an indicator for whether unit $i$ is already treated by period $t$. \citet[Lemma 5]{athey_design-based_2022} show that the coefficient on $A_{it}$ from the two-way fixed effects specification
\begin{equation}
Y_{it} = \alpha_i + \lambda_t + A_{it} \theta^{TWFE} + \epsilon_{it} \label{eqn: static twfe}
\end{equation}
\noindent can be decomposed as
\begin{equation}
\thetahat^{TWFE} = \sum_{t} \sum_{\substack{(g,g'): \\ min(g,g') \leq t}} \gamma_{t,gg'} \hat{\tau}_{t,gg'} + \sum_t \sum_{\substack{(g,g'): \\ min(g,g') > t}} \tilde\gamma_{t,gg'} \hat{\tau}_{t,gg'}
\end{equation}
\noindent for weights $\gamma_{t,gg'},\tilde\gamma_{t,gg'}$ that depend only on the $N_g$ and thus are non-stochastic in our framework. Thus, $\thetahat^{TWFE}$ can be viewed as an estimator of the form (\ref{eqn: defn of thetahatbeta}) for the parameter $\theta^{TWFE} = \sum_{t} \sum_{(g,g'): min(g,g') \leq t} \gamma_{t,gg'} \tau_{t,gg'}$, with $\hat{X} = -\sum_t \sum_{(g,g'): min(g,g') > t} \tilde\gamma_{t,gg'} \hat{\tau}_{t,gg'}$ and $\beta = 1$. As noted in \citet{athey_design-based_2022} and other papers, however, the parameter $\theta^{TWFE}$ may be difficult to interpret under treatment effect heterogeneity, since the weights $\gamma_{t,gg'}$ are not guided by economic reasoning, and moreover, some of the $\gamma_{t,gg'}$ may be negative, so that $\theta^{TWFE}$ is not a convex-weighted average of causal effects.
\end{example}

We note that, in principle, one can also use a vector-valued $\hat{X}$ that stacks the $\hat{X}$ values used by multiple estimators. For example, one could set $\hat{X} = (\hat{X}^{CS},\hat{X}^{CS2})'$, which combines the scalar values of $\hat{X}$ used for the two variants of the \citet{callaway_difference--differences_2020} estimator. Then $\hat\theta_{(1,0)'}$ would correspond to $\hat{\tau}^{CS}$, while $\hat\theta_{(0,1)'}$ would correspond to $\hat{\tau}^{CS2}$, thus nesting both estimators in the class of estimators of the form $\hat\theta_{\beta}$. One could likewise stack the $\hat{X}$'s associated with other DiD-related estimators. We stress, though, that the notion of efficiency that we derive below will be for the class of estimators using a specific vector $\hat{X}$; see Remarks \ref{rem: connection to semiparametrics} and \ref{rem: choice of xhat - theory} below for additional discussion.

\subsection{Efficient ``Oracle'' Estimation}
We now consider the problem of finding the best estimator $\thetahat_\beta$ of the form introduced in (\ref{eqn: defn of thetahatbeta}). We first show that $\thetahat_\beta$ is unbiased for all $\beta$, and then solve for the $\beta^*$ that minimizes the variance. 
\paragraph{Notation.} We begin by introducing some notation that will be useful for presenting our results. Recall that the sample treatment effect estimates $\tauhat_{t,gg'}$ are themselves differences in sample means, $\tauhat_{t,gg'} = \bar{Y}_{tg} - \bar{Y}_{tg'}$. It follows that we can write
$$ \thetahat_0 = \sum_g A_{\theta,g} \bar{Y}_g \text{ and } \hat X = \sum_g A_{0,g} \bar{Y}_g$$
\noindent for appropriately defined matrices $A_{\theta,g}$ and $A_{0,g}$ of dimension $1 \times T$ and $M \times T$, respectively, where $\bar{Y}_g = (\bar{Y}_{1g},...,\bar{Y}_{Tg})'$. Additionally, let $S_g =\varfin{Y_{i}(g)}$ be the finite population variance of $Y_i(g)$ and let $S_{gg'} = ({N-1})^{-1}  \sum_i  (Y_i(g) - \expefin{Y_i(g)}) (Y_i(g') - \expefin{Y_i(g')})'$ be the finite-population covariance between $Y_i(g)$ and $Y_i(g')$. \\

Our first result is that all estimators of the form $\thetahat_\beta$ are unbiased, regardless of $\beta.$
\begin{lem}[$\thetahat_\beta$ unbiased] \label{lem: thetahat beta unbiased}
Under Assumptions \ref{asm: random treatment} and \ref{asm: no anticipation},  $\expe{\thetahat_\beta} = \theta$ for any $\beta \in \reals^M$.
\end{lem}
\noindent See Remark \ref{rem: bias under nonrandom timing} below for a discussion of the bias that arises when Assumption \ref{asm: random treatment} fails.

We next turn our attention to finding the value $\beta^*$ that minimizes the variance.

\begin{prop} \label{prop: efficient betastar and variances}
Under Assumptions \ref{asm: random treatment} and \ref{asm: no anticipation}, the variance of $\thetahat_\beta$ is uniquely minimized at
\begin{equation}
\beta^* = \var{\hat{X}}^{-1} \cov{\hat{X}, \thetahat_0},   \label{eqn: expression for betastar regression} 
\end{equation}
provided that $\var{\hat{X}}$ is positive definite. Further, the variances and covariances in the expression for $\beta^*$ are given by
\begin{small}
\begin{align*}
\var{\twovec{\thetahat_0}{\hat X}} &= \twobytwomat{ \sum_g {N_g}^{-1} \, A_{\theta,g} \, S_g \, A_{\theta,g}' - N^{-1} S_\theta,  }{ \sum_g {N_g}^{-1} \, A_{\theta,g} \, S_g \, A_{0,g}'  }{ \sum_g {N_g}^{-1} \, A_{0,g} \, S_g \, A_{\theta,g}',}{ \sum_g {N_g}^{-1} \, A_{0,g} \, S_g \, A_{0,g}' } 
&=: \twobytwomat{V_{\thetahat_0}}{V_{\thetahat_0,\hat{X}}}{ V_{\hat{X},\thetahat_0} }{ V_{\hat{X}} },
\end{align*}
\end{small}
\noindent where $S_\theta = \varfin{ \sum_g A_{\theta,g} Y_i(g) }$. The efficient estimator has variance given by $\var{\thetahat_{\beta^*}} = V_{\thetahat_0} - (\beta^*)' V_{\hat{X}}^{-1} (\beta^*)$. 
\end{prop}

\noindent Equation (\ref{eqn: expression for betastar regression}) shows that the variance-minimizing $\beta^*$ is the best linear predictor of $\thetahat_0$ given $\hat{X}$. This formalizes the intuition that it is efficient to place more weight on pre-treatment differences in outcomes the more strongly they correlate with the post-treatment differences in outcomes. 

\begin{example}[continues=example:2periods]
In our ongoing two-period example, the efficient estimator $\thetahat_{\beta^*}$ derived in Proposition \ref{prop: efficient betastar and variances} is equivalent to the efficient estimator for cross-sectional randomized experiments in \citet{lin_agnostic_2013} and \citet{li_general_2017}. The optimal coefficient $\beta^*$ is equal to $\frac{N_\infty}{N} \beta_2 + \frac{N_2}{N} \beta_\infty$, where $\beta_g$ is the coefficient on $Y_{i1}$ from a regression of $Y_{i2}(g)$ on $Y_{i1}$ and a constant. Intuitively, this estimator puts more weight on the pre-treatment outcomes (i.e., $\beta^*$ is larger) the more predictive is the first period outcome $Y_{i1}$ of the second period potential outcomes. In the special case where the coefficients on lagged outcomes are equal to 1, the canonical difference-in-differences (DiD) estimator is optimal, whereas the simple difference-in-means (DiM) is optimal when the coefficients on lagged outcome are zero. For values of $\beta^* \in (0,1)$, the efficient estimator can be viewed as a weighted average of the DiD and DiM estimators.  
\end{example}

\subsection{Properties of the plug-in estimator}
Proposition \ref{prop: efficient betastar and variances} solves for the $\beta^*$ that minimizes the variance of $\thetahat_\beta$. However, the efficient estimator $\thetahat_{\beta*}$ is not of practical use since the ``oracle'' coefficient $\beta^*$ depends on the covariances of the potential outcomes, $S_g$, which are typically not known in practice. Mirroring \citet{lin_agnostic_2013} for cross-sectional randomized experiments, we now show that $\beta^*$ can be approximated by a plug-in estimate $\betahat^*$, and the resulting estimator $\thetahat_{\hat\beta^*}$ has similar properties to the ``oracle'' estimator $\thetahat_{\beta^*}$ when $N$ is large.

\subsubsection{Definition of the plug-in estimator}
To formally define the plug-in estimator, let 
$$\hat{S}_g = \frac{1}{N_g - 1} \sum_i  D_{ig} (Y_i(g) - \bar{Y}_g) (Y_i(g) - \bar{Y}_g)'$$ be the sample analog to $S_g$, and let $\hat{V}_{\hat{X},\thetahat_0}$ and $\hat{V}_{\hat{X}}$ be the analogs to $V_{\hat{X},\thetahat_0}$ and $V_{\hat{X}}$ that replace $S_g$ with $\hat{S}_g$ in the definitions. We then define the plug-in coefficient $$\betahat^* = \hat{V}_{\hat{X}}^{-1} \hat{V}_{\hat{X},\thetahat_0},$$
and consider the properties of the plug-in efficient estimator $\thetahat_{\betahat^*}$.

\begin{example}[continues=example:2periods]
In our ongoing two-period example, which we have shown is analogous to a cross-sectional randomized experiment, the plug-in estimator $\thetahat_{\betahat^*}$ is equivalent to the efficient plug-in estimator for cross-sectional experiments considered in \citet{lin_agnostic_2013}. As in \citet{lin_agnostic_2013}, $\thetahat_{\betahat^*}$ can be represented as the coefficient on $D_i$ in the interacted ordinary least squares (OLS) regression,
\begin{equation}
Y_{i2} = \beta_0 + \beta_1 D_i + \beta_2 \dot{Y}_{i1} + \beta_3 D_i \times \dot{Y}_{i1} + \epsilon_{i}, \label{eqn: interacted regression a la lin}
\end{equation}

\noindent where $\dot{Y}_{i1}$ is the demeaned value of $Y_{i1}$.\footnote{We are not aware of a representation of the plug-in efficient estimator as the coefficient from an OLS regression in the more general, staggered case.} Intuitively, this fully-interacted specification fits one linear model to estimate the mean of $Y_{i2}(2)$ and another to estimate $Y_{i2}(\infty)$, and then computes the difference, and thus is an augmented inverse propensity weighted (AIPW) estimator with a linear model for the conditional expectation functions and a constant propensity score \citep{glynn_introduction_2010}. 
\end{example}

\begin{Remark}[Connection to \citet{mckenzie_beyond_2012}] \label{rem: connection to mckenzie}
\citet{mckenzie_beyond_2012} proposes using an estimator similar to the plug-in efficient estimator in the two-period setting considered in our ongoing example. Building on results in \citet{frison_repeated_1992}, he proposes using the coefficient $\gamma_1$ from the OLS regression
\begin{equation}
Y_{i2} = \gamma_0 + \gamma_1 D_i + \gamma_2 \dot{Y}_{i1} + \epsilon_{i} \label{eqn: noninteracted regression a la lin},
\end{equation}
\noindent which is sometimes referred to as the Analysis of Covariance (ANCOVA I). This differs from the regression representation of the efficient plug-in estimator in (\ref{eqn: interacted regression a la lin}), sometimes referred to as ANCOVA II, in that it omits the interaction term $D_i \dot{Y}_{i1}$. Treating $\dot{Y}_{i1}$ as a fixed pre-treatment covariate, the coefficient $\gammahat_1$ from (\ref{eqn: noninteracted regression a la lin}) is equivalent to the estimator studied in \citet{Freedman(2008)-several_treatments,Freedman(2008)-regadj_to_experimental_data}. The results in \citet{lin_agnostic_2013} therefore imply that \citet{mckenzie_beyond_2012}'s estimator will have the same asymptotic efficiency as $\thetahat_{\betahat^*}$ under constant treatment effects. Intuitively, this is because the coefficient on the interaction term in (\ref{eqn: interacted regression a la lin}) converges in probability to 0. However, the results in \citet{Freedman(2008)-several_treatments, Freedman(2008)-regadj_to_experimental_data} imply that under heterogeneous treatment effects \citet{mckenzie_beyond_2012}'s estimator may even be less efficient than the simple difference-in-means $\thetahat_{0}$, which in turn is (weakly) less efficient than $\thetahat_{\betahat^*}$.\footnote{Relatedly, \citet{yang_efficiency_2001}, \citet{funatogawa_analysis_2011}, \citet{wan_analyzing_2020}, and \citet{negi_revisiting_2021} show that $\betahat_1$ from (\ref{eqn: interacted regression a la lin}) is asymptotically at least as efficient as $\gammahat_1$ from (\ref{eqn: noninteracted regression a la lin}) in sampling-based models similar to our ongoing example.}
\end{Remark}

\subsubsection{Asymptotic properties of the plug-in estimator}
We will now show that in large populations, the plug-in efficient estimator $\thetahat_{\betahat^*}$ is asymptotically unbiased for $\theta$ and has the same asymptotic variance as the oracle estimator $\thetahat_{\beta^*}$. To derive the properties of the plug-in efficient estimator in large finite populations, we consider a sequence of finite populations of increasing sizes, as in \citet{lin_agnostic_2013} and \citet{li_general_2017}, among other papers. More formally, we consider sequences of populations indexed by $m$ where the number of observations first treated at $g$, $N_{g,m}$, diverges for all $g \in \mathcal{G}$. For ease of notation, as in the aforementioned papers we leave the index $m$ implicit in our notation for the remainder of the paper. We assume the sequence of populations satisfies the following regularity conditions.

\begin{asm} \label{asm: regularity conditions for clt}
\begin{enumerate}[(i)]
    \item
    For all $g \in \mathcal{G}$, $N_g / N \rightarrow p_g \in (0,1)$.
    \item 
    For all $g,g'$, $S_g$ and $S_{gg'}$ have limiting values denoted $S^*_g$ and $S^*_{gg'}$, respectively, with $S^*_{g}$ positive definite.
    \item
    $\max_{i,g} || Y_i(g) - \expefin{Y_i(g)}||^2 / N \rightarrow 0$.
\end{enumerate}
\end{asm}

\noindent Part (i) imposes that the fraction of units first treated at period $g\in \mathcal{G}$ converges to a constant bounded between 0 and 1. Part (ii) requires the variances and covariances of the potential outcomes converge to a constant. Part (iii) requires that no single observation dominates the finite-population variance of the potential outcomes, and is thus analogous to the familiar Lindeberg condition in sampling contexts. 

With these assumptions in hand, we are able to formally characterize the asymptotic distribution of the plug-in efficient estimator. The following result shows that $\thetahat_{\betahat^*}$ is asymptotically unbiased and normally distributed, with the same asympototic variance as the ``oracle'' efficient estimator $\thetahat_{\beta^*}$. The proof exploits the general finite population central limit theorem in \citet{li_general_2017}.

\begin{prop} \label{prop: asymptotic dist of thetahat betahat}
Under Assumptions \ref{asm: random treatment}, \ref{asm: no anticipation}, and \ref{asm: regularity conditions for clt}, $$\sqrt{N}(\thetahat_{\betahat^*} - \theta) \rightarrow_d \normnot{0}{\sigma_*^2}, \hspace{1cm} \text{where} \hspace{1cm} \sigma_*^2 = \lim_{N\rightarrow\infty} N \var{\thetahat_{\beta^*}}.$$ 
\end{prop}

\begin{Remark}[Connection to semi-parametric efficiency] \label{rem: connection to semiparametrics}
Proposition \ref{prop: asymptotic dist of thetahat betahat} shows that the plug-in estimator $\hat\theta_{\hat\beta^*}$ achieves the same asymptotic variance as $\hat\theta_{\beta^*}$, the most efficient estimator in the class $\hat\theta_\beta$. We note that the asymptotic variance of the best estimator in this class is distinct from the semi-parametric efficiency bound in super-population frameworks \citep[e.g.][]{hahn_role_1998, santanna_doubly_2020}. We are not aware of any results on semi-parametric efficiency in design-based frameworks such as ours, nor are we aware of any results on the semi-parametric efficiency bound for panel data settings with staggered treatment timing---although both of these strike us as interesting directions for future research. Existing results do suggest a connection between our notion of efficiency and semi-parametric efficiency in our ongoing two-period example, however. \citet{negi_revisiting_2021} study covariate adjustment in cross-sectional randomized experiments from a super-population perspective, and show that \citet{lin_agnostic_2013}'s estimator (which they refer to as full regression adjustment, FRA) achieves the semi-parametric efficiency bound when the conditional expectation of the potential outcomes is linear in the observed covariates. Since our estimator is equal to FRA in our running two-period example (viewing $Y_{i1}$ as the pre-treatment covariate), this implies that $\hat\theta_{\hat\beta^*}$ is semi-parametric efficient (from the super-population perspective) when the conditional expectation of the second-period potential outcomes are linear in the pre-treatment outcome.
\end{Remark}

\begin{Remark}[On the choice of $\hat{X}$] \label{rem: choice of xhat - theory}
We note that $\var{\hat\theta_{\beta^*}} = V_{\hat\theta_0} - (\beta^*)' V_{X}^{-1} (\beta^*)$ can be viewed as the variance of the residual after linearly projecting $\hat\theta_0$ onto $\hat{X}$. Thus, the asymptotic variance of the plug-in efficient estimator will be smaller if $\hat{X}$ is more predictive of the estimation error in the simple difference-in-means estimator $\hat{\theta}_0$. It thus may be tempting to set $\hat{X}$ to be a vector including all possible comparisons of cohorts in periods before they were treated in order to minimize the asymptotic variance. This may not improve finite-sample performance, however, since the asymptotics considered in Proposition \ref{prop: asymptotic dist of thetahat betahat} assume that the number of observations $N$ is substantially larger than the dimension of $\hat{X}$, and thus may not approximate the finite-sample performance when $dim(\hat{X})$ is large. In particular, using a too high-dimensional $\hat{X}$ may lead to an ``overfitting'' problem analogous to controlling for too many pre-treatment variables in a cross-sectional experiment (see our Monte Carlo section below for an example of this phenomenon). \citet{lei_regression_2020} study covariate adjustment with a diverging number of covariates in cross-sectional randomized experiments, and find that (under certain regularity conditions), linear covariate adjustment works well when the dimension of the covariates is small relative to $N^{-\frac{1}{2}}$.\footnote{Future work might also consider an estimator that uses a high-dimensional $\hat{X}$, but considers some form of regularization on the coefficient $\hat\beta$. \label{fn: regularization}} We suspect a similar heuristic applies to the choice of the dimension of $\hat{X}$, although leave a formal analysis under diverging covariates to future work. In our Monte Carlo simulations below, we find good performance for the scalar $\hat{X}$ such that $\beta=1$ corresponds to the \citet[][]{callaway_difference--differences_2020} estimator, and thus consider this a reasonable default for practitioners implementing our method.
\end{Remark}

\begin{Remark}[Bias under non-random timing] \label{rem: bias under nonrandom timing}
Lemma \ref{lem: thetahat beta unbiased} shows that the oracle efficient estimator $\hat\theta_{\beta^*}$ is unbiased under Random Treatment Timing and No Anticipation. If, however, the Random Treatment Timing assumption is violated, then the efficient estimator may be biased, whereas the DiD estimator $(\beta=1)$ may still be unbiased under a parallel trends assumption. We note that $\hat\theta_{\beta^*} - \hat\theta_{1} = (\beta^*-1) \hat{X},$ and thus $ \expe{\hat\theta_{\beta^*} - \hat\theta_1} = (\beta^*-1) \cdot \expe{\hat{X}}$ (assuming $\hat{X}$ is scalar for simplicity). Hence, when DiD is unbiased but Random Treatment Timing is violated, the bias of the oracle efficient estimator will be larger (i) the farther is $\beta^*$ from 1, and (ii) the larger is $\expe{\hat{X}}$, i.e. the more imbalance there is in the pre-treatment outcome. We note, however, that when treatment timing is non-random, the DiD estimator will often be biased as well, e.g. when treatment is randomly assigned conditional on lagged outcomes \citep{AngristPischke(09), ding_bracketing_2019}.\footnote{As noted above, in the simple two period example, the plug-in efficient estimator is equivalent to an AIPW estimator with a constant propensity score and linear model for the conditional expectation function. Thus, from a super-population perspective, the plug-in efficient estimator would be consistent under the conditional unconfoundedness assumption, $1[G_i = 2] \perp Y_{i2}(.) | Y_{i1}$, when the conditional expectation functions are linear \citep[see, e.g.,][]{hahn_role_1998}. Formalizing this type of robustness in our design-based framework and extending it to settings with staggered treatment timing strikes us an interesting direction for future work.\label{fn: aipw}}    
\end{Remark}

\subsection{Inference \label{subsec: inference}}

We now introduce two methods for inference on $\theta$, the first using conventional $t$-based confidence intervals, and the second using Fisher randomization tests.

\subsubsection{$t$-based Confidence Intervals}

To construct confidence intervals using the asymptotic normal distribution derived in Proposition \ref{prop: asymptotic dist of thetahat betahat}, one requires an estimate of the variance $\sigma_*^2$. We first show that a simple Neyman-style variance estimator is conservative under treatment effect heterogeneity, as is common in finite population settings. We then introduce a less-conservative refinement to this estimator that adjusts for the part of the heterogeneity explained by $\hat{X}$.

Recall that $\sigma_*^2 = \lim_{N\rightarrow\infty} N \var{\thetahat_{\beta^*}}$. Examining the expression for $\var{\thetahat_{\beta^*}}$ given in Proposition \ref{prop: efficient betastar and variances}, we see that all of the components of the variance can be replaced with sample analogs except for the $-S_\theta$ term. This term corresponds with the variance of treatment effects, and is not consistently estimable since it depends on covariances between potential outcomes under treatments $g$ and $g'$ that are never observed simultaneously. This motivates the use of the Neyman-style variance that ignores the $-S_\theta$ term and replaces the variances $S_g$ with their sample analogs $\hat{S}_g$, \begin{small}$$\hat{\sigma}_*^2 = \left(\sum_g \frac{N}{N_g} \, A_{\theta,g} \, \hat{S}_g \, A_{\theta,g}' \right) - \left( \sum_g \frac{N}{N_g} \, A_{\theta,g} \, \hat{S}_g \, A_{0,g}' \right) \left( \sum_g \frac{N}{N_g} \, A_{0,g} \, \hat{S}_g \, A_{0,g}' \right)^{-1} \left( \sum_g \frac{N}{N_g} \, A_{\theta,g} \, \hat{S}_g \, A_{0,g}' \right)' .$$\end{small}

\noindent Since $\hat{S}_g \rightarrow_p S^*_g$ (see Lemma \ref{lem: variance consistency}), it is immediate that the estimator $\hat{\sigma}_*^2$ converges to an upper bound on the asymptotic variance $\sigma_*^2$, although the upper bound is conservative if there are heterogeneous treatment effects such that $S^*_\theta = \lim_{N\rightarrow \infty} S_\theta > 0$.

\begin{lem} \label{lem: consistency of sigmastarhat}
Under Assumptions \ref{asm: random treatment}, \ref{asm: no anticipation}, and \ref{asm: regularity conditions for clt}, $\hat{\sigma}_*^2 \rightarrow_p \sigma_*^2 + S_\theta^* \geq \sigma_*^2$.
\end{lem}

The estimator $\hat{\sigma}_*^2$ can be improved by using outcomes from earlier periods. The refined estimator intuitively lower bounds the heterogeneity in treatment effects by the part of the heterogeneity that is explained by the outcomes in earlier periods. The construction of this refined estimator mirrors the refinements using fixed covariates in randomized experiments considered in \citet{lin_agnostic_2013} and \citet{ abadie_sampling-based_2020}, with lagged outcomes playing a similar role to the fixed covariates. To avoid technical clutter, we defer the construction of the refined variance estimator to Appendix \ref{appendix: derivation of refined variance}, and merely state the sense in which the refined estimator improves upon the Neyman-style estimator introduced above. 

\begin{lem} \label{lem: consistency of sigmahat** - main text}
The refined estimator $\sigmahat_{**}$, defined in Lemma \ref{lem: consistency of sigmahat**}, satisfies $\hat \sigma_{**}^2 \rightarrow_p \sigma_*^2 + S^*_{\thetatilde}$, where $0 \leq S^*_{\thetatilde} \leq S^*_\theta$, so that $\hat \sigma_{**}$ is asymptotically (weakly) less conservative than $\hat \sigma_*$.
\end{lem}

\noindent It is then immediate that the confidence interval, $CI_{**} =\betahat^* \pm z_{1-\alpha/2} \cdot \widehat{se}$ is a valid $1-\alpha$ level confidence interval for $\theta$, where $\widehat{se} = \sigmahat_{**}/\sqrt{n}$ is the standard error and $z_{1-\alpha/2}$ is the $1-\alpha/2$ quantile of the normal distribution.

\subsubsection{Fisher Randomization Tests\label{subsec: frt}}

An alternative approach to inference uses Fisher randomization tests (FRTs), otherwise known as permutation tests. We will show that an FRT using a studentized version of the efficient estimator has the dual advantages that it 1) has exact size under the sharp null of no treatment effects for all units, and 2) is asymptotically valid for the weak null that $\theta = 0$. 

To derive the FRT, recall that the observed data is $(Y,G)$, where $Y$ collects all of the $Y_{it}$ and $G = (G_1,...,G_N)'$. Let $\mathcal{T} = \mathcal{T}(Y,G)$ denote a statistic of the data, and let $\mathcal{T}_\pi = \mathcal{T}(Y,G_\pi)$ be the statistic using the transformed data in which $G$ is replaced with a permutation $G_{\pi}$.\footnote{Formally, a permutation $\pi$ is a bijective map from $\{1,...,N\}$ onto itself, and $G_\pi = (G_{\pi(1)},...,G_{\pi(N)})'.$} A Fisher randomization test (FRT) computes the $p$-value
$$p_{FRT} = P_{\pi \sim U(\Pi)}( \mathcal{T}_\pi \geq \mathcal{T}(Y,G) ),$$
where the probability is taken over the uniform distribution on the set of permutations $\Pi$.\footnote{It is often difficult to calculate the $p$-value over all permutations exactly, so the $p$-value is approximated via simulation. We use 500 simulation draws in our simulations and 5,000 draws in the empirical application.} Under the sharp null hypothesis that $Y_{i}(g) = Y_{i}(g')$ for all $i,g,g'$, the distribution of $\mathcal{T}_\pi$ is the same as the distribution as $\mathcal{T}(Y,G)$, and thus by standard arguments the FRT is exact in finite samples (see, e.g., \citet{imbens_causal_2015}). 

The sharp null hypothesis of no treatment effect will often be too restrictive in practice, however, as we may be more interested in the hypothesis that the \textit{average} effect is zero, i.e., $H_0: \theta =0$. Unfortunately, in general FRTs may not have correct size for such weak null hypotheses even asymptotically \citep{wu_randomization_2020}.

We now show, however, that when the FRT is based on the studentized statistic $\mathcal{T}(Y,G) = \thetahat_{\betahat^*}/\widehat{se}$, it has asymptotically correct size under the weak null. In fact, we will show that asymptotically the FRT is equivalent to testing that $0$ falls within the $t$-based confidence interval $CI_{**}$ derived in the previous section. Thus, this FRT based on the studentized statistic is in some sense the ``best of both worlds'' of Fisherian and Neymanian inference in that it has exact size under the sharp null hypothesis while having asymptotically correct size under the weak null. 

The following regularity condition imposes that the means of the potential outcomes have limits, and that their fourth moment is bounded.
\begin{asm}
Suppose that for all $g$, $\lim_{N\rightarrow\infty} \expefin{Y_i(g)} = \mu_g < \infty$, and there exists $L < \infty$ such that $N^{-1} \sum_i ||Y_i(g) - \expefin{Y_i(g)}||^4 < L$ for all $N$. \label{asm: additional assumptions for clt}  
\end{asm}

With this assumption in hand, we can make precise the sense in which the FRT is asymptotically valid under the weak null.
\begin{prop} \label{prop: asymptotic properties of frt}
Suppose Assumptions \ref{asm: random treatment}-\ref{asm: additional assumptions for clt} hold. Let $t_\pi =(\thetahat^* / \widehat{se})_\pi$ be the studentized $t$-statistic under permutation $\pi$. Then $t_\pi \rightarrow_d \normnot{0}{1}$, $P_G$-almost surely. Hence, if $p_{FRT}$ is the $p$-value from the FRT associated with $|t_\pi|$, then under $H_0: \theta =0$, 
$$\lim_{N\rightarrow \infty} P(p_{FRT} \leq \alpha) \leq \alpha ,$$

\noindent $P_G$-almost surely, with equality if and only if $S_\theta^* = 0$.


\end{prop}

Proposition \ref{prop: asymptotic properties of frt} implies that the FRT using the studentized version of the efficient estimator asymptotically controls size under the weak null of no average treatment effects. Indeed, the proposition implies that the FRT is asymptotically equivalent to the test that the $t$-based confidence interval $CI_{**}$ includes $0$. Proposition \ref{prop: asymptotic properties of frt} extends the results in \citet[][]{wu_randomization_2020} and \citet[][]{zhao_covariate-adjusted_2020}, who consider permutation tests based on a studentized statistic in cross-sectional randomized experiments.\footnote{Permutation tests based on a studentized statistic have been considered in other contexts as well, for example \citet{janssen_studentized_1997,chung_exact_2013,chung_multivariate_2016,diciccio_robust_2017,bugni_inference_2018, MacKinnon2020,bai_inference_2022}.} Given the desirable properties of the FRT under both the sharp and weak null hypotheses, we recommend that researchers report $p$-values from the FRT alongside the usual $t$-based confidence intervals.

\subsection{Implications for existing estimators\label{subsec:examples}}
We now discuss the implications of our results for estimators previously proposed in the literature. We have shown that in the simple two-period case considered in Example \ref{example:2periods}, the canonical difference-in-differences corresponds with $\thetahat_1$. Likewise, in the staggered case, we showed in Examples \ref{example:callway and santa'anna}-\ref{example: dechaisemartin} that the estimators of \citet{callaway_difference--differences_2020}, \citet{sun_estimating_2020}, and \citet{de_chaisemartin_two-way_2020} correspond with the estimator $\thetahat_1$ for an appropriately defined estimand and $\hat{X}$. Our results thus imply that, unless $\beta^* =1$, the estimator $\thetahat_{\beta^*}$ is unbiased for the same estimand and has strictly lower variance under (quasi-)random treatment timing. Since the optimal $\beta^*$ depends on the potential outcomes, we do not generically expect $\beta^* =1$, and thus the previously-proposed estimators will generically be dominated in terms of efficiency. Although the optimal $\beta^*$ will typically not be known, our results imply that the plug-in estimator $\thetahat_{\betahat^*}$ will have similar properties in large populations, and thus will be more efficient than the previously-proposed estimators in large populations under (quasi-)random treatment timing. We thus recommend the plug-in efficient estimator in settings where parallel trends is justified with random treatment timing.

We note, however, that the estimators in the aforementioned papers are valid for the ATT in settings where only parallel trends holds but there is not random treatment timing, whereas the validity of the efficient estimator depends on random treatment timing (see Remark \ref{rem: pt vs random timing} above).\footnote{The estimator of \citet{de_chaisemartin_two-way_2020} can also be applied in settings where treatment turns on and off over time.} Although in some settings parallel trends is justified by arguing that treatment is (quasi-)randomly assigned, in some observational settings the researcher may be more comfortable imposing parallel trends than quasi-random treatment timing. We thus view the the plug-in efficient estimator to be complementary to the estimators considered in previous work, since it is more efficient under stricter assumptions that will not hold in all cases of interest.


\subsection{Extensions and practical considerations\label{subsec: extensions}}

We now discuss several extensions and practical considerations that may be useful for applying our methods.

\begin{Remark}[Testing the randomization assumption]
It may often be desirable to test the assumption of (quasi-)random treatment timing, especially in non-experimental settings where random timing cannot be ensured by design. We briefly describe three approaches. First, since consistency of the efficient estimator depends on the assumption that $\expe{\hat X} =0$, a natural falsification test is to test whether $\hat X$ is significantly different from zero --- i.e. are there significant differences in pre-treatment means between cohorts treated at different times. It is straightforward to conduct a test of the null that $\expe{\hat X} =0$ using a one-sample $t$-test (with a sample analog to the variance given in Proposition \ref{prop: efficient betastar and variances}) or using an FRT. Second, an intuitive approach which mirrors the common practice of testing for pre-existing trends is to estimate an event-study, treating the initial time of treatment as $G_i - k$ for some $k>0$, and then test whether the dynamic effects corresponding with the leads $1,...,k$ are different from zero.\footnote{Given that the efficient estimator differs from the usual DiD estimator, note that this test differs from the common pre-test for pre-existing trends.}  Third, as is common in randomized controlled trials, researchers can test for covariate balance between units treated at different times. For example, \citet{deshpande_who_2019} show that observable characteristics do not predict the timing of social security office closings. We illustrate how these types of tests can be used in our application below. Such tests can be a useful test of the plausibility of the randomization assumption, and can help to identify cases where it is clearly violated.  We caution, however, that as with tests of pre-existing trends \citep[cf.][]{roth_pretest_2022}, such falsification tests may have limited power to detect violations of the randomization assumption, and relying on them can introduce distortions from pre-testing. Thus, it is best to additionally motivate the randomization assumption based on context-specific knowledge.
\end{Remark}

\begin{Remark}[Conditional Random Treatment Timing]
For simplicity, we have considered the case of unconditional random treatment timing. In some experiments, the treatment timing may be randomized among units with some shared observable characteristics (e.g. counties within a state). In this case, the methodology described above can be applied within each randomization stratum, and the stratum-level estimates can be pooled to form aggregate estimates for the population.\footnote{The FRTs can likewise be modified to consider permutations that permute assignments only within randomization strata.} Likewise, in quasi-experimental contexts, the assumption of quasi-random treatment timing may be more plausible among sub-groups of the population (e.g. within units of the same gender and education status), or among groups of units that were treated at similar times (e.g. within a decade). The units can then be partitioned into strata based on discrete observable characteristics, and the analysis we describe can be conducted within each stratum. Extending our results to allow for randomization conditional on a continuous characteristic is an interesting topic for future work.
\end{Remark}

\begin{Remark}[Clustered Treatment Assignment]
Likewise, in some settings there may be clustered assignment of treatment timing  --- e.g. treatment is assigned to families $f$, and all units $i$ in family $f$ are first treated at the same time. This violates Assumption \ref{asm: random treatment}, since not all vectors of treatment timing are equally likely. However, note that any average treatment contrast at the individual level, e.g. $\frac{1}{N} \sum_i Y_{it}(g) - Y_{it}(g')$, can be written as an average contrast of a transformed family-level outcome, e.g. $\frac{1}{F} \sum_f \tilde{Y}_{ft}(g) - \tilde{Y}_{ft}(g')$, where $\tilde{Y}_{ft}(g) = (F/N) \sum_{i \in f} Y_{it}(g)$. Thus, clustered assignment can easily be handled in our framework by analyzing the transformed data at the cluster level.
\end{Remark}

\begin{Remark}[Fixed pre-treatment covariates]
In some settings, researchers may also have access to fixed pre-treatment covariates $W_i$. Differences in the mean of $W_i$ between adoption cohorts can then be added to the vector $\hat{X}$ to further increase precision.
\end{Remark}

\section{Monte Carlo Results\label{sec: monte carlo}}

We present two sets of Monte Carlo results. In Section \ref{subsec: two period sims}, we conduct simulations in a stylized two-period setting matching our ongoing example to illustrate how the plug-in efficient estimator compares to the classical difference-in-differences and simple difference-in-means (DiM) estimators. Section \ref{subsec:wood et al sims} presents a more realistic set of simulations with staggered treatment timing that is calibrated to our application, comparing the plug-efficient estimator to recent DiD-based estimators proposed for the staggered treatment case.

\subsection{Two-period Simulations\label{subsec: two period sims}.}

\paragraph{Specification.} We follow the model in Example \ref{example:2periods} in which there are two periods ($t=1,2$) and units are treated in period two or never-treated $(\mathcal{G}= \{1,2\})$. We generate the potential outcomes as follows. For each unit $i$ in the population, we draw the never-treated potential outcomes $Y_i(\infty) = (Y_{i1}(\infty), Y_{i2}(\infty))'$ from a $\normnot{0}{\Sigma_\rho}$ distribution, where $\Sigma_\rho$ has 1s on the diagonal and $\rho$ on the off-diagonal. The parameter $\rho$ is the correlation between the untreated potential outcomes in period $t=1$ and period $t=2$. We then set $Y_{i2}(2) = Y_{i2}(\infty) + \tau_i$, where $\tau_i = \gamma (Y_{i2}(\infty) - \expefin{Y_{i2}(\infty)})$. The parameter $\gamma$ governs the degree of heterogeneity of treatment effects: if $\gamma =0$, then there is no treatment effect heterogeneity, whereas if $\gamma$ is positive then individuals with larger untreated outcomes in $t=2$ have larger treatment effects. We center by $\expefin{Y_{i2}(\infty)}$ so that the treatment effects are 0 on average. We generate the potential outcomes once, and treat the population as fixed throughout our simulations. Our simulation draws then differ based on the draw of the treatment assignment vector. For simplicity, we set $N_2 = N_\infty = N/2$, and in each simulation draw, we randomly select which units are treated in $t=1$ or not. We conduct 1000 simulations for all combinations of $N_2 \in \{25,1000\}$, $\rho \in \{0,.5,.99\},$ and $\gamma \in \{0,0.5\}$. 

\paragraph{Results.} Table \ref{tbl: 2-period-monte-carlos} shows the bias, standard deviation, and coverage of 95\% confidence intervals for the plug-in efficient estimator $\thetahat_{\betahat^*}$, difference-in-differences $\thetahat^{DiD} = \thetahat_1$, and simple differences-in-means $\thetahat^{DiM} = \thetahat_0$. It also shows the size (null rejection probability) of the FRT using a studentized statistic introduced in Section \ref{subsec: inference}. Confidence intervals are constructed as $\thetahat_{\betahat^*} \pm 1.96 \sigmahat_{**}/\sqrt{n}$ for the plug-in efficient estimator, and analogously for the other estimators.\footnote{\label{fn:variance estimator for other estimators}For $\thetahat_\beta$, we use an analog to $\hat \sigma_{**}$, except the unrefined estimate $\hat{\sigma}_*$ is replaced with the sample analog to the expression for $\var{\thetahat_\beta}$ implied by Proposition \ref{prop: efficient betastar and variances}.} For all specifications and estimators, the estimated bias is small, and coverage is close to the nominal level. Table \ref{tbl: 2-period-monte-carlos-ratios} facilitates comparison of the standard deviations of the different estimators by showing the ratio relative to the plug-in estimator. The standard deviation of the plug-in efficient estimator is weakly smaller than that of either DiD or DiM in nearly all cases, and is never more than 2\% larger than that of either DiD or DiM. The standard deviation of the plug-in efficient estimator is similar to DiD when auto-correlation of $Y(0)$ is high $(\rho=0.99)$ and there is no heterogeneity of treatment effects $(\gamma=0)$, so that $\beta^* \approx 1$ and thus DiD is (nearly) optimal in the class we consider. Likewise, it is similar to DiM when there is no autocorrelation $(\rho=0)$ and there is no treatment effect heterogeneity $(\gamma=0)$, and thus $\beta^* \approx 0$ and so DiM is (nearly) optimal in the class we consider. The plug-in efficient estimator is substantially more precise than DiD and DiM in many other specifications: the standard deviation of DiD can be as much as 1.7 times larger than the plug-in efficient estimator, and the standard deviation of the DiM can be as much as 7 times larger. These simulations thus illustrate how the plug-in efficient estimator can improve on DiD or DiM in cases where they are suboptimal, while retaining nearly identical performance when the DiD or DiM model is optimal.

\begin{table}[hbtp]
\caption{Bias, Standard Deviation, and Coverage for $\thetahat_{\betahat^*}$, $\thetahat^{DiD}, \thetahat^{DiM}$ in 2-period simulations}
\resizebox{\columnwidth}{!}{
\captionsetup[table]{labelformat=empty,skip=1pt}

\begin{tabular}{rrrrrrrrrrrrrrrr}
\toprule
& & & & \multicolumn{3}{c}{Bias} & \multicolumn{3}{c}{SD} & \multicolumn{3}{c}{Coverage} & \multicolumn{3}{c}{FRT Size} \\ 
 \cmidrule(lr){5-7}\cmidrule(lr){8-10}\cmidrule(lr){11-13}\cmidrule(lr){14-16}
$N_1$ & $N_0$ & $\rho$ & $\gamma$ & PlugIn & DiD & DiM & PlugIn & DiD & DiM & PlugIn & DiD & DiM & PlugIn & DiD & DiM \\ 
\midrule
1000 & 1000 & 0.99 & 0.0 & $0.00$ & $0.00$ & $-0.00$ & $0.01$ & $0.01$ & $0.04$ & $0.95$ & $0.95$ & $0.95$ & $0.05$ & $0.05$ & $0.05$ \\ 
1000 & 1000 & 0.99 & 0.5 & $0.00$ & $0.00$ & $-0.00$ & $0.01$ & $0.01$ & $0.06$ & $0.95$ & $0.95$ & $0.95$ & $0.04$ & $0.06$ & $0.05$ \\ 
1000 & 1000 & 0.50 & 0.0 & $0.00$ & $0.00$ & $0.00$ & $0.04$ & $0.04$ & $0.05$ & $0.94$ & $0.95$ & $0.94$ & $0.06$ & $0.05$ & $0.05$ \\ 
1000 & 1000 & 0.50 & 0.5 & $0.00$ & $0.00$ & $0.00$ & $0.05$ & $0.05$ & $0.06$ & $0.95$ & $0.95$ & $0.95$ & $0.06$ & $0.05$ & $0.05$ \\ 
1000 & 1000 & 0.00 & 0.0 & $-0.00$ & $0.00$ & $-0.00$ & $0.04$ & $0.07$ & $0.04$ & $0.95$ & $0.94$ & $0.95$ & $0.05$ & $0.06$ & $0.05$ \\ 
1000 & 1000 & 0.00 & 0.5 & $-0.00$ & $0.00$ & $-0.00$ & $0.06$ & $0.07$ & $0.06$ & $0.95$ & $0.95$ & $0.95$ & $0.04$ & $0.05$ & $0.05$ \\ 
25 & 25 & 0.99 & 0.0 & $0.00$ & $0.00$ & $-0.03$ & $0.04$ & $0.04$ & $0.27$ & $0.94$ & $0.94$ & $0.94$ & $0.04$ & $0.05$ & $0.06$ \\ 
25 & 25 & 0.99 & 0.5 & $0.00$ & $-0.01$ & $-0.04$ & $0.05$ & $0.08$ & $0.34$ & $0.92$ & $0.93$ & $0.93$ & $0.06$ & $0.06$ & $0.06$ \\ 
25 & 25 & 0.50 & 0.0 & $-0.01$ & $0.02$ & $-0.02$ & $0.24$ & $0.29$ & $0.26$ & $0.94$ & $0.95$ & $0.94$ & $0.04$ & $0.04$ & $0.05$ \\ 
25 & 25 & 0.50 & 0.5 & $-0.01$ & $0.01$ & $-0.03$ & $0.30$ & $0.32$ & $0.33$ & $0.94$ & $0.95$ & $0.94$ & $0.04$ & $0.04$ & $0.05$ \\ 
25 & 25 & 0.00 & 0.0 & $-0.03$ & $-0.02$ & $-0.03$ & $0.28$ & $0.38$ & $0.27$ & $0.93$ & $0.95$ & $0.93$ & $0.06$ & $0.04$ & $0.06$ \\ 
25 & 25 & 0.00 & 0.5 & $-0.04$ & $-0.02$ & $-0.04$ & $0.35$ & $0.42$ & $0.34$ & $0.93$ & $0.94$ & $0.94$ & $0.06$ & $0.05$ & $0.06$ \\ 
\bottomrule
\end{tabular}

}
\label{tbl: 2-period-monte-carlos}
\end{table} 

\begin{table}[!hbtp]
\caption{Ratio of standard deviations for $\thetahat^{DiD}$ and $\thetahat^{DiM}$ relative to $\thetahat_{\betahat^*}$ in 2-period simulations}
\captionsetup[table]{labelformat=empty,skip=1pt}
\begin{longtable}{rrrrrrrr}
\toprule
& & & & & \multicolumn{3}{c}{SD Relative to Plug-In} \\ 
 \cmidrule(lr){6-8}
$N_1$ & $N_0$ & $\rho$ & $\gamma$ & $\beta^*$ & PlugIn & DiD & DiM \\ 
\midrule
1000 & 1000 & 0.99 & 0.0 & $0.99$ & $1.00$ & $1.00$ & $7.09$ \\ 
1000 & 1000 & 0.99 & 0.5 & $1.24$ & $1.00$ & $1.71$ & $7.07$ \\ 
1000 & 1000 & 0.50 & 0.0 & $0.52$ & $1.00$ & $1.13$ & $1.15$ \\ 
1000 & 1000 & 0.50 & 0.5 & $0.65$ & $1.00$ & $1.04$ & $1.15$ \\ 
1000 & 1000 & 0.00 & 0.0 & $-0.03$ & $1.00$ & $1.45$ & $1.00$ \\ 
1000 & 1000 & 0.00 & 0.5 & $-0.03$ & $1.00$ & $1.31$ & $1.00$ \\ 
25 & 25 & 0.99 & 0.0 & $0.97$ & $1.00$ & $0.99$ & $6.58$ \\ 
25 & 25 & 0.99 & 0.5 & $1.22$ & $1.00$ & $1.47$ & $6.31$ \\ 
25 & 25 & 0.50 & 0.0 & $0.41$ & $1.00$ & $1.21$ & $1.10$ \\ 
25 & 25 & 0.50 & 0.5 & $0.51$ & $1.00$ & $1.08$ & $1.10$ \\ 
25 & 25 & 0.00 & 0.0 & $0.10$ & $1.00$ & $1.35$ & $0.98$ \\ 
25 & 25 & 0.00 & 0.5 & $0.13$ & $1.00$ & $1.22$ & $0.98$ \\ 
\bottomrule
\end{longtable}

\label{tbl: 2-period-monte-carlos-ratios}
\end{table} \setcounter{table}{\thetable -2}

\subsection{Simulations Based on \citet{wood_reanalysis_2020}\label{subsec:wood et al sims}}
To evaluate the performance of our proposed methods in a more realistic staggered setting, we conduct simulations calibrated to our application in Section \ref{sec: wood application}, which is based on data from \citet{wood_reanalysis_2020}. The outcome of interest $Y_{it}$ is the number of complaints against police officer $i$ in month $t$ for police officers in Chicago. Police officers were randomly assigned to first receive a procedural justice training in period $G_i$. See Section \ref{sec: wood application} for more background on the application.

\paragraph{Simulation specification.} We calibrate our baseline specification as follows. The number of observations and time periods in the data exactly matches that used in our application. We set the untreated potential outcomes $Y_{it}(\infty)$ to match the observed outcomes in the data $Y_{it}$ (which would exactly match the true potential outcomes if there were no treatment effect on any units). In our baseline simulation specification, there is no causal effect of treatment, so that $Y_{it}(g) = Y_{it}(\infty)$ for all $g$. (We describe an alternative simulation design with heterogeneous treatment effects in Appendix Section \ref{appendix section: additional sims}.) In each simulation draw $s$, we randomly draw a vector of treatment dates $G_s = (G_1^s,...,G_N^s)'$ such that the number of units first treated in period $g$ matches that observed in the data (i.e. $\sum 1[G_i^s=g] = N_g$ for all $g$). In total, there are 72 months of data on 5537 officers. There are 47 distinct values of $g$, with the cohort size $N_g$ ranging from 3 to 575. In an alternative specification, we collapse the data to the yearly level, so that there are 6 time periods and 5 larger cohorts. 

For each simulated data-set, we calculate the plug-in efficient estimator $\thetahat_{\betahat^*}$ for four estimands: the simple-weighted average treatment effect $(\theta^{simple})$; the calendar- and cohort-weighted average treatment effects ($\theta^{calendar}$ and $\theta^{cohort}$), and the instantaneous event-study parameter $(\theta^{ES}_0)$.\footnote{We do not report results for the estimand of TWFE specifications, in light of the recent literature showing that these estimands do not have an intuitive causal interpretation in settings with staggered treatment timing (e.g. \citet{borusyak_revisiting_2016,athey_design-based_2022, goodman-bacon_difference--differences_2018, de_chaisemartin_two-way_2020,sun_estimating_2020}). The results for the DiD estimator in the previous section illustrate the performance of TWFE in a simple setting where it has an intuitive estimand.} (See Section \ref{subsec: target parameter} for the formal definition of these estimands). In our baseline specification, we use as $\hat{X}$ the scalar weighted combination of pre-treatment differences used by the \citet[][CS]{callaway_difference--differences_2020} estimator using not-yet-treated units as the comparison ($\hat\tau^{CS2}$ in Example \ref{example:callway and santa'anna}). In the appendix, we also present results for an alternative specification in which $\hat{X}$ is a vector containing $\tauhat_{t,gg'}$ for all pairs $g,g'>t$. For comparison, we also compute the CS and \citet[][SA]{sun_estimating_2020} estimators for the same estimand. Recall that for $\theta_0^{ES}$, the CS estimator coincides with the estimator proposed in \citet{de_chaisemartin_two-way_2020} in our setting, since treatment is an absorbing state. Confidence intervals are calculated as $\thetahat_{\betahat^*} \pm 1.96 \sigmahat_{**}/\sqrt{n} $ for the plug-in efficient estimator and analogously for the CS and SA estimators.\footnote{The variance estimator for the CS and SA estimators is adapted analogously to that for the DiD and DiM estimators, as discussed in footnote \ref{fn:variance estimator for other estimators}. We note that these design-based standard errors differ slightly from those proposed in the original CS and SA papers, which adopt a sampling-based framework; using design-based standard errors makes the CIs for these estimators more directly comparable to those for the plug-in efficient estimator.}

\paragraph{Baseline simulation results.} The results for our baseline specification are shown in Tables \ref{tbl: simulation results - main spec} and \ref{tbl: comparison of SDs - main spec}. As seen in Table \ref{tbl: simulation results - main spec}, the plug-in efficient estimator is approximately unbiased, and 95\% confidence intervals based on our standard errors have coverage rates close to the nominal level for all of the estimands, with size distortions no larger than 3\% for all of our specifications. The size for the FRT is also close to the nominal level, which is intuitive since our baseline specification imposes the sharp null hypothesis, and thus the FRT should be exact up to simulation error. The CS and SA estimators are also both approximately unbiased and have coverage close to the nominal level, although coverage for the SA estimator is as low as 90\% in some specifications.

\begin{table}[!htb]
    \centering
    \captionsetup[table]{labelformat=empty,skip=1pt}
\begin{longtable}{llrrrrr}
\toprule
Estimator & Estimand & Bias & Coverage & FRT Size & Mean SE & SD \\ 
\midrule
PlugIn & calendar & 0.01 & 0.93 & 0.07 & 0.26 & 0.28 \\ 
PlugIn & cohort & 0.00 & 0.92 & 0.06 & 0.26 & 0.28 \\ 
PlugIn & ES0 & 0.00 & 0.96 & 0.04 & 0.32 & 0.31 \\ 
PlugIn & simple & 0.00 & 0.93 & 0.05 & 0.24 & 0.25 \\ 
CS & calendar & 0.01 & 0.95 & 0.06 & 0.51 & 0.52 \\ 
CS & cohort & 0.02 & 0.95 & 0.04 & 0.47 & 0.46 \\ 
CS/dCDH & ES0 & 0.00 & 0.96 & 0.04 & 0.44 & 0.43 \\ 
CS & simple & 0.02 & 0.96 & 0.04 & 0.47 & 0.46 \\ 
SA & calendar & 0.00 & 0.91 & 0.04 & 1.44 & 1.50 \\ 
SA & cohort & 0.01 & 0.90 & 0.05 & 1.51 & 1.58 \\ 
SA & ES0 & 0.00 & 0.96 & 0.04 & 0.91 & 0.94 \\ 
SA & simple & 0.02 & 0.90 & 0.05 & 1.64 & 1.72 \\ 
 \bottomrule
\end{longtable}

    \caption{Results for Simulations Calibrated to \citet{wood_reanalysis_2020}}
    \label{tbl: simulation results - main spec}
    \floatfoot{Note: This table shows results for the plug-in efficient and CS and SA estimators in simulations calibrated to \citet{wood_reanalysis_2020}. The estimands considered are the calendar-, cohort-, and simple-weighted average treatment effects, as well as the instantaneous event-study effect (ES0). The CS estimator for ES0 corresponds with the estimator in \citet{de_chaisemartin_two-way_2020}. Coverage refers to the fraction of the time a nominal 95\% confidence interval includes the true parameter, and FRT size refers to the null rejection rate of a Fisher Randomization Test. Mean SE refers to the average estimated standard error, and SD refers to the actual standard deviation of the estimator. The bias, Mean SE, and SD are all multiplied by 100 for ease of readability.}
\end{table} \setcounter{table}{\thetable -1}

\begin{table}[!htb]
    \centering
    \captionsetup[table]{labelformat=empty,skip=1pt}
\begin{longtable}{lrr}
\toprule
 & \multicolumn{2}{c}{Ratio of SD to Plug-In} \\ 
 \cmidrule(lr){2-3}
Estimand & CS & SA \\ 
\midrule
calendar & $1.84$ & $5.31$ \\ 
cohort & $1.67$ & $5.72$ \\ 
ES0 & $1.39$ & $3.02$ \\ 
simple & $1.85$ & $6.86$ \\ 
 \bottomrule
\end{longtable}

    \caption{Comparison of Standard Deviations -- \citet{callaway_difference--differences_2020} and \citet{sun_estimating_2020} versus Plug-in Efficient Estimator}
    \label{tbl: comparison of SDs - main spec}
    \floatfoot{Note: This table shows the ratio of the standard deviation of the CS and SA estimators relative to the plug-in efficient estimator, based on the simulation results in Table \ref{tbl: simulation results - main spec}.}
\end{table}

Table \ref{tbl: comparison of SDs - main spec} shows that there are large efficiency gains from using the plug-in efficient estimator relative to the CS or SA estimators. The table compares the standard deviation of the plug-in efficient estimator to that of the CS and SA estimators. Remarkably, using the plug-in efficient estimator reduces the standard deviation relative to the CS estimator by a factor between 1.39 and 1.85, depending on the estimand. Since standard errors are proportional to the square root of the sample size for a fixed estimator, a reduction in standard errors by a factor of 1.85 roughly corresponds with an increase in sample size by a factor of 3.4. The gains of using the plug-in efficient estimator relative to the SA estimator are even larger, with reductions in the standard deviation by a factor of three or more. The reason for this is that the SA estimator uses only the last-treated units (rather than not-yet-treated units) as a comparison, but in our setting less than 1\% of units are treated in the final period, leading to an efficiency loss.  

\paragraph{Alternative choices of $\hat{X}$.} In Appendix \ref{appendix section: additional sims}, we present results where $\hat{X}$ is set to be a vector containing all possible comparisons of cohorts in periods prior to treatment. The dimension of this $\hat{X}$ is large relative to $N$ using monthly data, and in line with the discussion in Remark \ref{rem: choice of xhat - theory}, we find that the estimator has large bias and undercoverage owing to an over-fitting problem. When collapsing the data to the yearly level, the dimension of the $\hat{X}$ is more moderate, and the estimator is approximately unbiased and has good coverage, in line with the heuristic from \citet{lei_regression_2020} that the dimension of $\hat{X}$ should be small relative to $N^{-\frac{1}{2}}$. In our Monte Carlo simulation, however, the augmented $\hat{X}$ offers very minor precision gains relative to the $\hat{X}$ based on the \citet{callaway_difference--differences_2020} estimator used in our baseline specification.\footnote{\label{fn: xhat with lags}We also experimented with a 10-dimensional vector $\hat{X}$ that included our baseline scalar choice of $\hat{X}$ as the first element, as well as analogs to the baseline choice lagged by $1,...,9$ periods, with very similar results to the baseline specification.} We thus focus on the latter choice of $\hat{X}$ in our application below.

\paragraph{Other Extensions.} Appendix \ref{appendix section: additional sims} contains several extensions to the baseline simulation specification, such as incorporating heterogeneous effects, annualizing the monthly data, and considering the other two outcomes in our application. As in the baseline specification, the plug-in efficient estimator has good coverage and offers efficiency gains relative to the other methods in nearly all specifications.

\section{Application to Procedural Justice Training\label{sec: wood application}}
\subsection{Background}
Reducing police misconduct and use of force is an important policy objective. \citet{wood_procedural_2020} studied the Chicago Police Department's staggered rollout of a procedural justice training program, which taught police officers strategies for emphasizing respect, neutrality, and transparency in the exercise of authority. Officers were randomly assigned a date for training. \citet{wood_procedural_2020} found large and statistically significant impacts of the program on complaints and sustained complaints against police officers and on officer use of force. However, our re-analysis in \citet{wood_reanalysis_2020} highlighted a statistical error in the original analysis of \citet{wood_procedural_2020}, which failed to normalize for the fact that groups of officers trained in different months were of varying sizes. In \citet{wood_reanalysis_2020}, we re-analyzed the data using the procedure proposed by \citet{callaway_difference--differences_2020} to correct for the error. The re-analysis found no significant effect on complaints or sustained complaints, and borderline significant effects on use of force, although the confidence intervals for all three outcomes included both near-zero and meaningfully large effects. \citet{owens_can_2018} studied a small pilot study of a procedural justice training program in Seattle, with point estimates suggesting reductions in complaints but imprecisely estimated.

\subsection{Data}
We use the same data as in the re-analysis in \citet{wood_reanalysis_2020}, which extends the data used in the original analysis of \citet{wood_procedural_2020} through December 2016. As in \citet{wood_reanalysis_2020}, we restrict attention to the balanced panel of officers who remained in the police force throughout the study period. We further drop officers in the initial pilot program and who are in special units, as these officers were trained in large batches and did not follow the random assignment protocol (see the supplementary material to \citet{wood_procedural_2020}). This leaves us a final sample of 5537 officers.\footnote{In the earlier working paper version of this paper, \citet{roth_efficient_2021}, we included officers in the pilot program and special units, with qualitatively similar results. However, one can formally reject the null hypothesis of random assignment when including these officers (see Table \ref{tbl:covariate balance}).} The data contain three outcome measures (complaints, sustained complaints, and use of force) at a monthly level for 72 months (6 years), with the first cohort trained in month 17 and the final cohort trained in the last month of the sample.

\subsection{Estimation}

We apply our proposed plug-in efficient estimator to estimate the effects of the procedural justice training program on the three outcomes of interest. As in our Monte Carlo study, we use the scalar $\hat{X}$ such that $\beta=0$ is the \citet[][]{callaway_difference--differences_2020} estimator ($\tauhat^{CS2}$). We estimate the simple, cohort, and calendar-weighted average effects described in Section \ref{subsec: target parameter} and used in our Monte Carlo study. We also estimate the event-study effects for the first 24 months after treatment, which includes the instantaneous event-study effect studied in our Monte Carlo as a special case (for event-time 0). For comparison, we also estimate the \citet{callaway_difference--differences_2020} estimator as in \citet{wood_reanalysis_2020}.\footnote{The CS estimates are not identical to those in \citet{wood_reanalysis_2020} for two reasons (although are qualitatively similar). The first is that we exclude officers in the pilot program and special units. Second, for direct comparability, we calculate design-based standard errors for the CS estimator using the analog to $\sigmahat_{**}$, and thus the reported SEs differ slightly from the sampling-based SEs reported in \citet{wood_reanalysis_2020}.}

\subsection{Results}

\paragraph{Baseline results.} Figure \ref{fig: wood-et-al application summary comparison} shows the results of our analysis for the three aggregate summary parameters. Table \ref{tbl: wood-et-al application percentages table} compares the magnitudes of these estimates and their 95\% confidence intervals (CIs) to the mean of the outcome in the 12 months before the pilot program began. It also reports $p$-values from the FRT.

\begin{figure}[!h]
    \centering
    \includegraphics[width=0.9\linewidth]{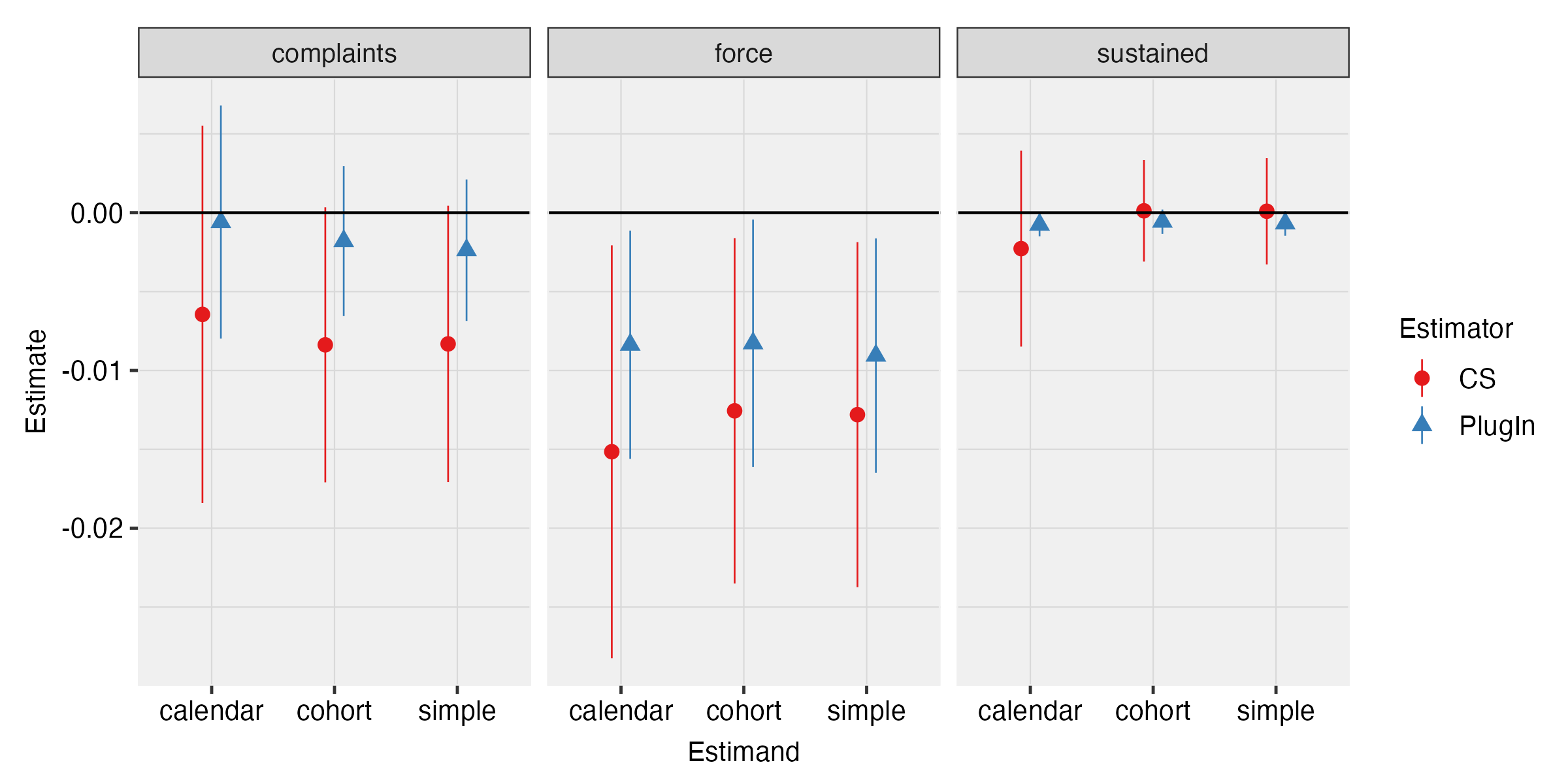}
    \caption{Effect of Procedural Justice Training Using the Plug-In Efficient and \citet{callaway_difference--differences_2020} Estimators} \label{fig: wood-et-al application summary comparison}
    \floatfoot{Note: this figure shows point estimates and 95\% CIs for the effects of procedural justice training on complaints, force, and sustained complaints using the CS and plug-in efficient estimators. Results are shown for the calendar-, cohort-, and simple-weighted averages.}
\end{figure}

\begin{table}[!hbt]
    \centering
    \includegraphics[width = 0.95\linewidth]{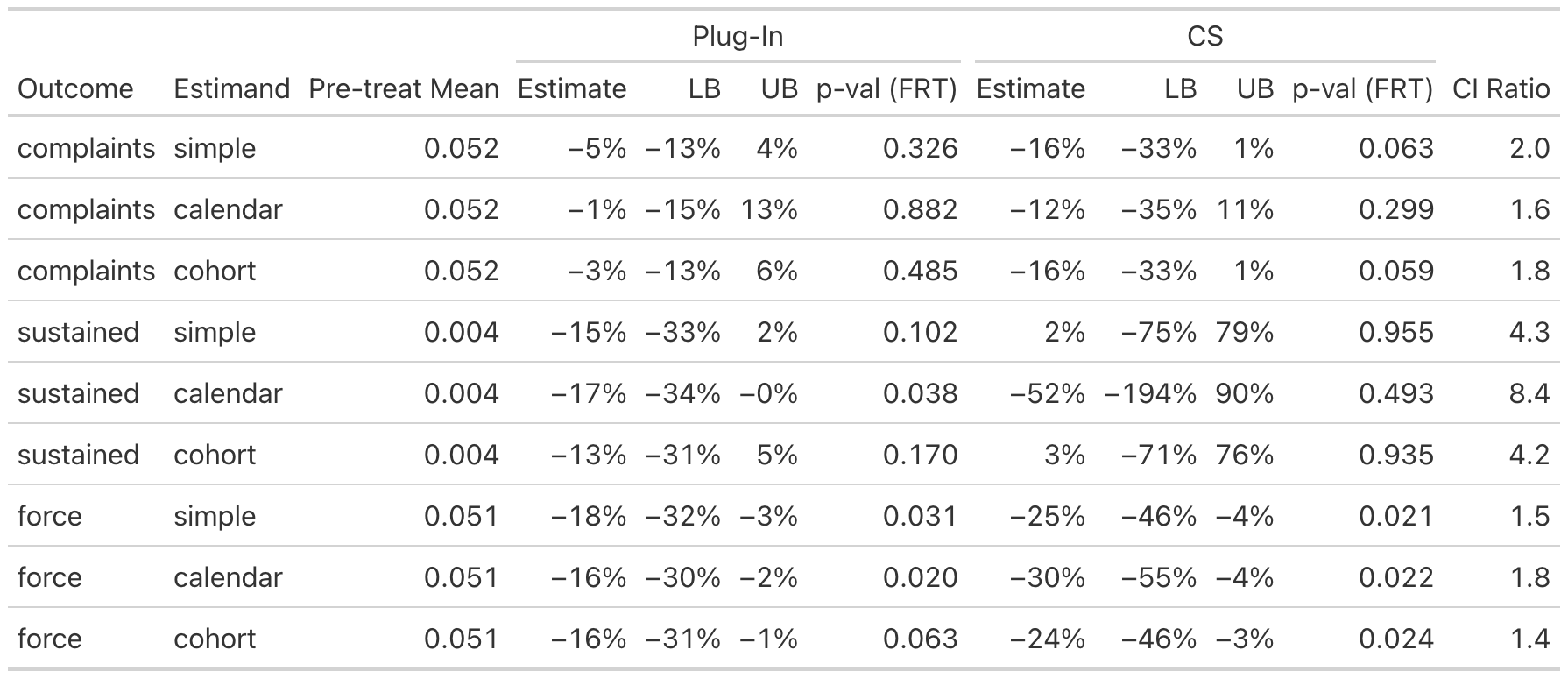}
    \caption{Estimates and 95\% CIs as a Percentage of Pre-treatment Means}\label{tbl: wood-et-al application percentages table}
    \floatfoot{Note: This table shows the pre-treatment means for the three outcomes. It also displays the estimates and 95\% CIs in Figure \ref{fig: wood-et-al application summary comparison} as percentages of these means, as well as the $p$-value from a Fisher Randomization Test (FRT). The final column shows the ratio of the length of the CI for the CS estimator relative to that for the plug-in efficient estimator. All FRT $p$-values are based on 5,000 permutations.}
\end{table} 

For all outcomes, the CIs for the plug-in efficient estimator overlap with those of the \citet[~CS]{callaway_difference--differences_2020} estimator but are substantially narrower. Indeed, the final column of Table \ref{tbl: wood-et-al application percentages table} shows that the standard errors (or equivalently, the length of the CIs) range from 1.4 to 8.4 times smaller depending on the specification. As in \citet{wood_reanalysis_2020}, we find no significant impact on complaints using any of the aggregations. Our bounds on the magnitude of the treatment effect are substantially tighter than before, however. For instance, using the simple aggregation we can now rule out reductions in complaints of more than 13\%, compared with a bound of 33\% using the CS estimator, and our standard errors are roughly twice as small as when using CS. For use of force, the point estimates from the efficient estimator are somewhat smaller (in magnitude) than using CS, but suggest a reduction in force of around 16 to 18 percent of the pre-treatment mean. However, the upper bounds of the confidence intervals are close to zero; $p$-values using the FRT are between 0.02 and 0.06. Thus, although precision is substantially higher than when using the CS estimator, the CIs for force still include effects from near-zero up to about 30\% of the pre-treatment mean. For sustained complaints, all of the point estimates are near zero and the CIs are substantially narrower than when using the CS estimator, although the plug-in efficient estimate using the calendar aggregation is marginally significant (FRT $p$-value $= 0.04$). In Appendix Figures \ref{fig: wood-et-al application event-study efficient}-\ref{fig: wood-et-al application event-study CS}, we show event-study plots using the plug-in efficient and CS estimates. The figures do not show a clear significant effect for any of the outcomes, nor do they show significant placebo pre-treatment effects.

\paragraph{Balance and robustness checks.} Although treatment timing was explicitly randomized in our application, as discussed in the supplement to \citet{wood_procedural_2020}, there are some concerns about non-compliance wherein officers could volunteer to receive the training before their randomly assigned date, particularly towards the end of the training period. (The observed treatment variable in the data is the \textit{actual} training date, and whether an officer volunteered is not recorded.) We therefore conduct a series of robustness and balance checks to evaluate the extent to which non-compliance may have violated the assumption of random treatment timing. We first test for balance in pre-treatment outcomes by testing the null that $\expe{\hat{X}} =0$, as described in Section \ref{subsec: extensions}. In particular, we use the scalar $\hat{X}$ used by the CS estimator for each of our summary parameters and outcomes, with results shown in Table \ref{tbl:covariate balance}. Reassuringly, we do not find any (individual or jointly) significant imbalances in $\hat{X}$ using our main analysis sample. Interestingly, we do find a significant imbalance for use of force if we include officers in the pilot program and special units, who are known not to have followed the randomization protocol, which suggests that these tests may be powered to detect some relevant violations of the randomization assumption. Second, we construct an ``event-study plot'' that tests for placebo pre-treatment effects prior to the date of training, and generally do not find any concerning pre-treatment placebo effects. These results, and those for our subsequent balance checks, are shown in Appendix \ref{appendix: additional application results}. Third, we test for covariate balance on year of birth, one of the few pre-treatment demographic variables in the data. We find that average year of birth is similar across training dates, although in one of our two specifications we statistically reject the null of exact equality at the 10\% level (sup-t $p$-value $=0.08$), possibly suggesting some slight imbalance on age. Finally, as a robustness check we re-do our main analysis excluding units who were trained in the final year of the training program, when noncompliance was suspected to be more severe. The qualitative patterns are similar, although the estimates for use of force are no longer statistically significant in some specifications.

\paragraph{Implications.} Our analysis provides the most precise estimates to date on the effectiveness of procedural justice training for police officers. Our estimates for the effects of the program on complaints against officers are close to zero, with much tighter upper bounds on the effectiveness at reducing complaints than in previous work. The results for force are more mixed, with point estimates suggesting reductions of 16-18\%, but confidence intervals that include near-zero or zero effects in all specifications. Thus, more research is needed to determine whether procedural justice training can be a useful tool in meaningfully reducing officer use of force. We encourage police departments planning to implement such trainings in the future to consider a randomized staggered rollout, which is a potentially low-cost way to learn more about the effectiveness of the program. 

\begin{table}[!ht]
    \centering
    \caption{Tests of balance on pre-treatment outcomes}
    \includegraphics[width = 0.9\textwidth]{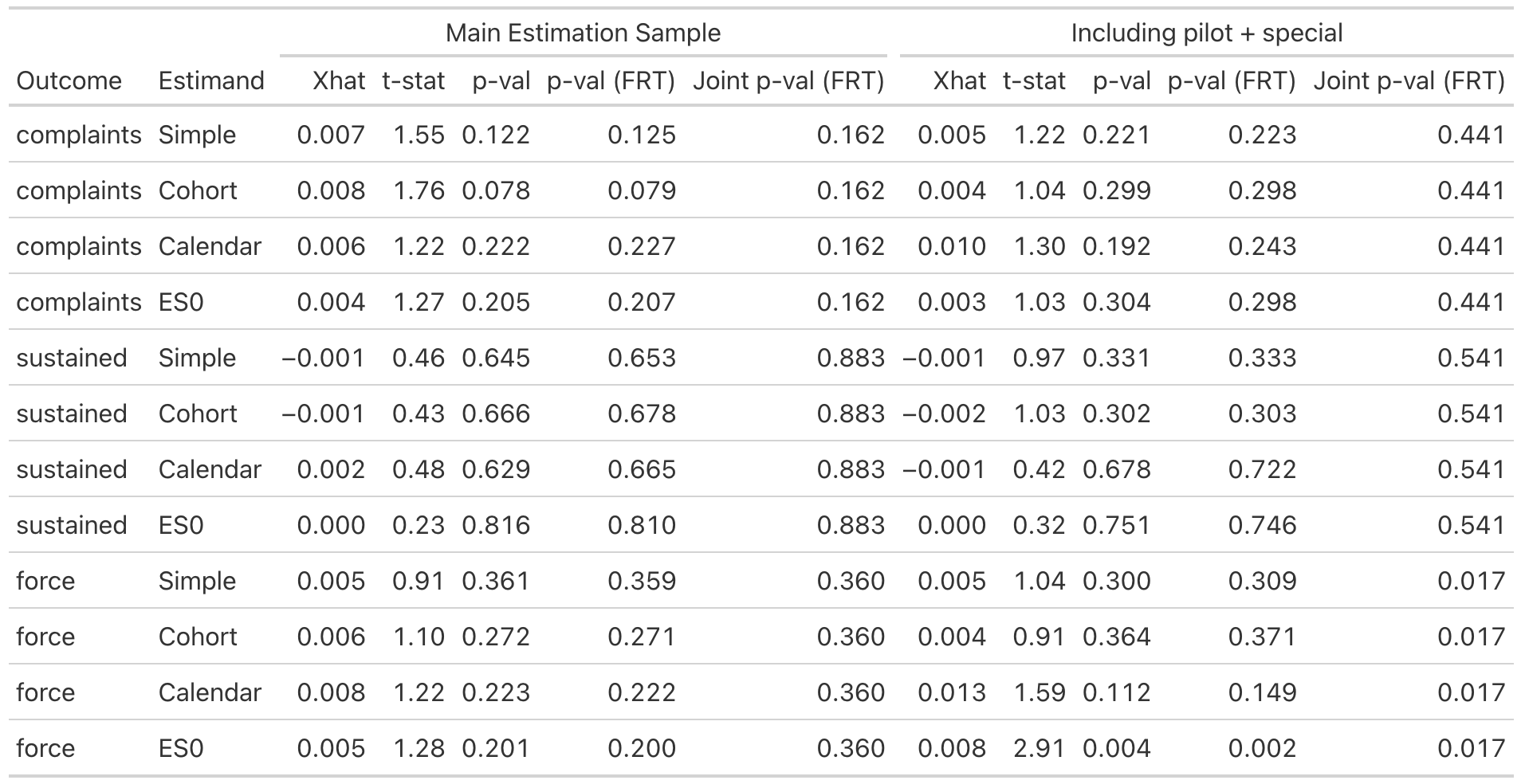}
    \label{tbl:covariate balance}
    \floatfoot{Note: this table shows balance on pre-treatment outcomes by testing the null hypothesis that $\expe{\hat{X}} =0$. The columns report the value of $\hat{X}$, its $t$-statistic, $p$-value based on the $t$-statistic, $p$-value using an FRT, and a $p$-value for the joint test that $\expe{\hat{X}}=0$ for all estimands using the same outcome (computed using an FRT with the $\max |t|$ statistic). The columns labeled Main Estimation Sample use the main data for our analysis, whereas those labeled ``Including pilot + special'' include officers in the pilot program and special units, who did not follow the randomization protocol. All FRT $p$-values are based on 5,000 permutations.}
\end{table}

\section{Conclusion}

This paper considers efficient estimation in settings with staggered adoption and (quasi-)~random treatment timing. The assumption of (quasi-)random treatment timing is technically stronger than parallel trends, but is often the justification given for the parallel trends assumption in practice, and it can be ensured by design in experimental contexts where the researcher controls the timing of treatment. We derive the most efficient estimator in a large class of estimators that nests many existing approaches. The ``oracle'' efficient estimator is not known in practice, but we show that a plug-in sample analog has similar properties in large populations, and we derive both $t$-based and permutation-based approaches to inference. We find in simulations that the proposed plug-in efficient estimator is approximately unbiased, yields reliable inference, and substantially increases precision relative to existing methods. We apply our proposed methodology to obtain the most precise estimates to date of the causal effects of procedural justice training programs for police officers.

\clearpage
\setlength{\baselineskip}{7.1mm}

\small{
\setlength{\bibsep}{1pt plus 0.3ex}
\putbib
}
\end{bibunit}
\begin{bibunit}

\appendix 
\clearpage


\paragraph{\Large{Supplement to ``Efficient Estimation for Staggered Rollout Designs''}}

\renewcommand{\tablename}{Appendix Table}
\renewcommand{\figurename}{Appendix Figure}
\setcounter{figure}{0}
\setcounter{table}{-1}

\section{Proofs}

\paragraph{Proof of Lemma \ref{lem: thetahat beta unbiased}}
\begin{proof}
By Assumption \ref{asm: random treatment}, $\expe{D_{ig}} = ({N_g}/{N})$. Hence,  
\begin{align*}
\expe{\thetahat_0} = \expe{\sum_g A_{\theta,g} \frac{1}{N_g} \sum_i  D_{ig} Y_{i} } = \sum_g A_{\theta,g} \frac{1}{N_g} \sum_i  \expe{D_{ig}} Y_i(g) 
= \sum_g A_{\theta,g} \frac{1}{N_g} \sum_i  \frac{N_g}{N} Y_i(g) = \theta.
\end{align*}
Likewise,
\begin{align*}
\expe{\hat{X}} = \expe{\sum_g A_{0,g} \frac{1}{N_g} \sum_i  D_{ig} Y_{i} }  
= \sum_g A_{0,g} \frac{1}{N} \sum_i  Y_i(g) = \frac{1}{N} \sum_i  \sum_g A_{0,g} Y_i(g) = 0,
\end{align*}
\noindent since $\sum_g A_{0,g} Y_i(g) = 0$ by Assumption \ref{asm: no anticipation}. The result follows immediately from the previous two displays.
\end{proof}

\paragraph{Proof of Proposition \ref{prop: efficient betastar and variances}}
\begin{proof}
First, observe that $$\min_\beta \var{\thetahat_\beta} = \min_\beta \var{\thetahat_0 - \hat{X}' \beta} = \min_\beta \expe{ \left( (\thetahat_0 - \theta) - (\hat{X} - \expe{\hat{X}})' \beta) \right)^2 }.$$ From the usual least-squares formula, the unique solution is $$ \underbrace{\expe{(\hat{X} - \expe{\hat{X}}) (\hat{X} - \expe{\hat{X}})'}^{-1}}_{\var{\hat{X}}^{-1}} \underbrace{\expe{(\hat{X} - \expe{\hat{X}}) (\thetahat_0 - \theta)}}_{\cov{\hat{X}, \thetahat_0}},$$ which gives the first result. 

To derive the form of the variance, let $A_{\tau,g} = \twovec{A_{\theta,g}}{A_{0,g}}$. Define $$\tauhat := \sum_g A_{\tau,g} \bar{Y}_g = \twovec{\thetahat_0}{\hat{X}}.$$ Since Assumption \ref{asm: random treatment} holds, we can appeal to Theorem 3 in \citet{li_general_2017}, which implies that $\var{\tauhat} = \sum_g {N_g}^{-1} A_{\tau,g} S_g A_{\tau,g}' - N^{-1} S_\tau$, where $S_\tau = \varfin{\sum_{g} A_{\tau,g} Y_i(g)}$. The result then follows immediately from expanding this variance, as well as the observation that $S_\tau = \twobytwomat{S_\theta}{0}{0}{0}$, where the 0 blocks are obtained by noting that $\sum_g A_{0,g} Y_i(g) =0$ for all $i$ by Assumption \ref{asm: no anticipation}.
\end{proof}

\paragraph{Proof of Proposition \ref{prop: asymptotic dist of thetahat betahat}}

To establish the proof, we first provide two lemmas that characterize the asymptotic joint distribution of $(\thetahat_0, \hat{X}')'$, and show that $\hat{S}_g$ is consistent for $S_g^*$ under Assumption \ref{asm: regularity conditions for clt}. Both results are direct consequences of the general asymptotic results in \citet{li_general_2017} for multi-valued treatments in randomized experiments.

\begin{lem}\label{lem: asymptotic distribution of thetahat xhat}
Under Assumptions \ref{asm: random treatment}, \ref{asm: no anticipation}, and \ref{asm: regularity conditions for clt},
$$\sqrt{N} \twovec{ \thetahat_0 - \theta}{ \hat{X} } \rightarrow_d \normnot{0}{V^*},$$ where $$V^* = \twobytwomat{ \sum_g {p_g}^{-1} \, A_{\theta,g} \, S_g^* \, A_{\theta,g}' - S_\theta^*  }{ \sum_g {p_g}^{-1} \, A_{\theta,g} \, S_g^* \, A_{0,g}'  }{ \sum_g {p_g}^{-1} \, A_{0,g} \, S_g^* \, A_{\theta,g}'}{ \sum_g {p_g}^{-1} \, A_{0,g} \, S_g^* \, A_{0,g}' } =: \twobytwomat{V^*_{\thetahat_0}}{ V^*_{\thetahat_0,\hat{X}} }{ V^*_{\hat{X},\thetahat_0} }{ V^*_{\hat{X}}},$$ and $S^*_\theta = \lim_{N\rightarrow\infty} S_\theta$ (where $S_\theta$ is defined in Proposition \ref{prop: efficient betastar and variances}).
\end{lem}

\begin{proof}
As in the proof to Proposition \ref{prop: efficient betastar and variances}, we can write $$\tauhat = \sum_g A_{\tau,g} \bar{Y}_g = \twovec{\thetahat_0}{\hat{X}}.$$ The result then follows from Theorem 5 in \citet{li_general_2017}, combined with the observation noted in the proof to Proposition \ref{prop: efficient betastar and variances} that $S_\tau = \twobytwomat{S_\theta}{0}{0}{0}$ and hence $S_\tau \rightarrow \twobytwomat{S^*_\theta}{0}{0}{0}$.
\end{proof}

\begin{lem} \label{lem: variance consistency}
Under Assumptions \ref{asm: random treatment}, \ref{asm: no anticipation}, and \ref{asm: regularity conditions for clt}, $\hat{S}_g \rightarrow_p S_g^*$ for all $g$.
\end{lem}

\begin{proof}
Follows immediately from Proposition 3 in \citet{li_general_2017}.
\end{proof}

To complete the proof of Proposition \ref{prop: efficient betastar and variances}, recall that $\betahat^* = \hat{V}_{\hat{X}}^{-1} \hat{V}_{\hat{X},\thetahat_0}$. It is clear that $\betahat^*$ is a continuous function of $\hat{V}_{\hat{X}}$ and $\hat{V}_{\hat{X},\thetahat_0}$, and that $\hat{V}_{\hat{X}}$ and $\hat{V}_{\hat{X},\thetahat_0}$ are continuous functions of $\hat{S}_g$. From Lemma \ref{lem: variance consistency} along with the continuous mapping theorem, we obtain that $\betahat^* \rightarrow_p (V^{*}_X)^{-1} V^*_{\hat{X},\thetahat_0}$. Lemma \ref{lem: asymptotic distribution of thetahat xhat} together with Slutsky's lemma then give that $\sqrt{N}(\thetahat_{\betahat^*} - \theta) \rightarrow_d \normnot{0}{ V^*_{\thetahat_0} - V^{*\prime}_{\hat{X},\thetahat_0} (V_{\hat{X}}^*)^{-1} V^*_{\hat{X},\thetahat_0}}$. From Proposition \ref{prop: efficient betastar and variances}, it is apparent that the asymptotic variance of $\thetahat_{\betahat^*}$ is equal to the limit of $N \var{\thetahat_{\beta^*}}$, which completes the proof.

\paragraph{Proof of Lemma \ref{lem: consistency of sigmastarhat}}
\begin{proof}
Immediate from the fact that $\hat{S}_g \rightarrow_p S^*_g$ (see Lemma \ref{lem: variance consistency}) combined with the continuous mapping theorem.
\end{proof}

\paragraph{Proof of Proposition \ref{prop: asymptotic properties of frt}}
\begin{proof}
Note that, conditional on $G$, the distribution of $t_\pi$ corresponds with the distribution of $\sqrt{N}\thetahat^* / \hat\sigma_{**}$ in a population with potential outcomes $Y^*(\cdot)$, where $Y_i^*(g) = Y_i(G_i)$ for all $i,g$. To prove the first assertion, it thus suffices to show that the populations defined by $Y^*(\cdot)$ satisfy Assumption \ref{asm: regularity conditions for clt}, $P_G$-almost surely, in which case the result follows from Proposition \ref{prop: asymptotic dist of thetahat betahat} and Lemma \ref{lem: consistency of sigmahat**} applied to the population with potential outcomes $Y^*(\cdot)$. 

Since the set of observations with $G_i =g$ is a simple random sample from a finite population, Lemma A5 in \citet{wu_randomization_2020} implies that 
$$ \bar{Y}_g = \frac{1}{N_g} \sum_i 1[G_i = g] Y_i(g) \rightarrow_{a.s.} \lim_{N\rightarrow \infty } \expefin{Y_i(g)} =: \mu_g^* $$

$$ \hat{S}_g = \frac{1}{N_g-1} \sum_i 1[G_i = g] (Y_i(g) - \bar{Y}_g)^2  \rightarrow_{a.s.} \lim_{N\rightarrow \infty } \varfin{Y_i(g)} =: S_g^*$$
In a slight abuse of notation, we will denote by $\expefin{Y_i^*(g)}$ the finite-population expectation in the population with potential outcomes $Y_i^*(g)$, where $Y_i^*(g) = Y_i(G_i)$. Now,
\begin{align*}
\expefin{Y_i^*(g)} = \frac{1}{N} \sum_i Y_i = \sum_g \frac{N_g}{N} \frac{1}{N_g}\sum_i 1[G_i =g] Y_i(g) \rightarrow_{a.s.} \sum_g p_g \mu_g^*     
\end{align*}
Similarly,
\begin{align*}
\varfin{Y_i^*(g)} &= \frac{1}{N-1} \sum_i (Y_i - \bar{Y})^2 \\&= \frac{N}{N-1} \left(  \left( \sum_g \frac{N_g}{N} \left( \frac{1}{N_g}\sum_i 1[G_i=g] Y_i^2 - \bar{Y}_g^2 \right) \right)  + \left( \sum_g \frac{N_g}{N} \bar{Y}_g^2 - \left(\sum_g \frac{N_g}{N} \bar{Y}_g \right)^2 \right) \right) \\
&= \frac{N}{N-1} \left(  \left( \sum_g \frac{N_g}{N} \frac{N_g-1}{N_g} \hat{S}_g \right) + \left( \sum_g \frac{N_g}{N} \bar{Y}_g^2 - \left(\sum_g \frac{N_g}{N} \bar{Y}_g \right)^2 \right) \right) \\
&\rightarrow_{a.s.} \sum_g p_g S_g^* + \left(\sum_g p_g (\mu_g^{*})^2 - \left(\sum_g p_g \mu_g^*\right)^2 \right)
\end{align*}
\noindent where we obtain the convergence from the previous displays and the continuous mapping theorem (and we use the shorthand $Y^2$ for $Y Y'$). The first term in the limit is positive definite, since $S_g^*$ is positive definite for each $g$ by Assumption \ref{asm: regularity conditions for clt}, and the second term is positive semi-definite (it is the variance of the discrete distribution with probability $p_g$ on $\mu_g$). Hence, Assumption \ref{asm: regularity conditions for clt}(ii) is satisfied for the population with potential outcomes $Y_i^*$ $P_g$-almost surely. Finally, Assumption \ref{asm: regularity conditions for clt}(iii) is satisfied $P_g$-almost surely by Lemma A6 in \citet{wu_randomization_2020}.

The second assertion then follows immediately from the fact that $\sqrt{N}(\thetahat_{\betahat^*} - \theta)/ \sigmahat_{**} \rightarrow_d \normnot{0}{c}$, for $c = \sigma_*^2/( \sigma_*^2 + S_\theta^*) \leq 1$, by Proposition \ref{prop: asymptotic dist of thetahat betahat} and Lemma \ref{lem: consistency of sigmahat**}. 
\end{proof}

\subsection{Derivation of Variance Refinement\label{appendix: derivation of refined variance}}

We now provide a derivation for the refined variance estimator discussed in Lemma \ref{lem: consistency of sigmahat** - main text}, as well as a formal proof of its validity. First, recall that the Neyman-style variance estimator was conservative by $S^*_\theta = \lim_{N\rightarrow\infty} S_\theta$. We first provide a lemma which gives a consistently estimable lower bound on $S_\theta$. Intuitively, this is the component of the treatment effect heterogeneity that is explained by lagged outcomes.

\begin{lem} \label{lem: decomp of s theta when there are unused periods }
Suppose that $A_{\theta,g}=0$ for all $g < g_{min}$. If Assumption \ref{asm: no anticipation} holds, then
\begin{equation}
S_\theta = \varfin{\tilde\theta_i} + \frac{N+1}{N-1} \left( \sum_{g \geq g_{min}} \beta_g \right)' \left( M S_{g_{min}} M' \right)  \left( \sum_{g \geq g_{min}} \beta_g \right), \label{eqn: decomp of s theta when there are unused periods }    
\end{equation}
\noindent where $M$ is the matrix that selects the rows of $Y_i$ corresponding with $t < g_{min}$; $\beta_g = (M S_{g} M')^{-1} M S_g A_{\theta,g}'$ is the coefficient from projecting $A_{\theta,g} Y_i(g)$ on $M Y_i(g)$ (and a constant);  and $\tilde\theta_i = \sum_{g\geq g_{min}} A_{\theta,g} Y_i(g) - \sum_{g \geq g_{min}} (M Y_i(g))' \beta_g$.
\end{lem}
\begin{proof}
For any $g$ and functions of the potential outcomes $X_i \in \reals^K$ and $Z_i \in \reals$, let $\dot{X}_i = X_i - \expefin{X_i}$, $\dot{Z}_i = Z_i - \expefin{Z_i}$, and $\beta_{XZ} =  \varfin{X_i}^{-1} \expefin{\dot{X}_i \dot{Z}_i}$. Observe that
\begin{align*}
\varfin{Z_i - \beta_{XZ}' X_i} &= \frac{1}{N-1} \sum_i  \left(\dot{Z}_i - \beta_{XZ}'\dot{X}_i \right)^2 \\
&= \frac{1}{N-1} \sum_i  \dot{Z}_i^2 +  \beta_{XZ}' \left( \frac{1}{N-1} \sum_i  \dot{X}_i \dot{X}_i'\right) \beta_{XZ} -\beta_{XZ}' \frac{2}{N-1} \sum_i  \dot{X_i} \dot{Z}_i  \\
&= \varfin{Z_i} + \beta_{XZ}' \varfin{X_i} \beta_{XZ} - 2 \frac{N}{N-1} \beta_{XZ}' \varfin{X_i} \beta_{XZ} \\
&= \varfin{Z_i} - \frac{N+1}{N-1} \beta_{XZ}' \varfin{X_i} \beta_{XZ}.
\end{align*}
The result then follows from setting $Z_i = \sum_{g\geq g_{min}} A_{\theta,g} Y_i(g) = \theta_i$ and $X_i = M Y_{i}(g_{min})$, and noting that under Assumption \ref{asm: no anticipation}, $MY_{i}(g_{min}) = M Y_i(g)$ for all $g\geq g_{min}$, and hence $\varfin{MY_{i}(g_{min})} = M S_{g_{min}} M' = M S_g M' = \varfin{MY_i(g)}$. 
\end{proof}

We now formally define the refined estimator $\sigmahat_{**}$ and give a more detailed statement of Lemma \ref{lem: consistency of sigmahat** - main text}.

\begin{lem} \label{lem: consistency of sigmahat**}
Suppose that $A_{\theta,g}=0$ for all $g < g_{min}$ and Assumptions \ref{asm: random treatment}-\ref{asm: regularity conditions for clt} hold. Let $M$ be the matrix that selects the rows of $Y_i$ corresponding with periods $t<g_{min}$. Define
\begin{equation}
\hat\sigma_{**}^2 = \hat\sigma_{*}^2 - \left( \sum_{g > g_{min}} \betahat_g \right)' \left( M \hat{S}_{g_{min}} M' \right)  \left( \sum_{g > g_{min}} \betahat_g \right),  \label{eqn: defn of sigmahat**}
\end{equation} where $\betahat_g = (M \hat{S}_g M')^{-1} M \hat{S}_g A_{\theta,g}' $. Then $\hat \sigma_{**}^2 \rightarrow_p \sigma_*^2 + S^*_{\thetatilde}$, where $0 \leq S^*_{\thetatilde} \leq S^*_\theta$, so that $\hat \sigma_{**}$ is asymptotically (weakly) less conservative than $\hat \sigma_*$.

\end{lem}

\paragraph{Proof of Lemma \ref{lem: consistency of sigmahat**}}
\begin{proof}
Note that $\betahat_g$ is a continuous function of $\hat{S}_g$. Lemma \ref{lem: variance consistency} together with the continuous mapping theorem thus imply that $$\left( \sum_{g > g_{min}} \betahat_g \right)' \left( M \hat{S}_{g_{min}} M' \right)  \left( \sum_{g > g_{min}} \betahat_g \right) - \left( \sum_{g > g_{min}} \beta_g \right)' \left( M S_{g_{min}} M' \right)  \left( \sum_{g > g_{min}} \beta_g \right) \rightarrow_p 0.$$ From Lemmas \ref{lem: consistency of sigmastarhat} and \ref{lem: decomp of s theta when there are unused periods }, it is then immediate that $\sigma_{**}^2 \rightarrow_p \sigma_*^2 + S^*_{\thetatilde}$, where $S^*_{\thetatilde} = \lim_{N\rightarrow\infty} \varfin{\thetatilde_i} \leq \lim_{N\rightarrow\infty} S_\theta = S_\theta^*$.
\end{proof}

\section{Additional Simulation Results\label{appendix section: additional sims}}

This section presents results from extensions to the simulations in Section \ref{sec: monte carlo}.

\paragraph{Other outcomes.} Appendix Tables \ref{tbl: simulation results - main spec - force}-\ref{tbl: comparison of SDs - main spec - sustained} show results analogous to those in the main text, except using the other two outcomes considered in our application (use of force and sustained complaints). We again find that the plug-in efficient estimator has minimal bias and is substantially more precise than the CS and SA estimators in all specifications (with reductions in standard deviations relative to CS by a factor of over 3 for some specifications). Likewise, both $t$-based and FRT-based approaches yield reliable inference in all specifications.\footnote{In an earlier version of our simulations, in which we included units in the pilot program, we did find some undercoverage (79\%) of $t$-based CIs for the plug-in efficient estimator for the calendar aggregation with sustained complaints. The distinguishing features of this specification were that the outcome is very rare (pre-treatment mean 0.004) and the aggregation scheme places the largest weight on the small number of units in the pilot cohort (which had only 17 officers). This does not appear to be an issue in our current simulations, where as in our application, we drop units in the pilot program, eliminating the small pilot cohorts. We thus urge some caution in applying the efficient estimator (or any approach based on a central limit theorem) in settings where one is placing substantial weight on small cohorts. In such cases, it is preferable to use the FRT (which is valid under the sharp null) or collapse the data to a more aggregated level to form larger cohorts.} 

\paragraph{Annualized data.} Appendix Tables \ref{tbl: simulation results - annual - complaints}-\ref{tbl: comparison of SDs - annual - sustained} present simulations from an alternative specification where the monthly data is collapsed to the yearly level, so that there are six total time periods and five (larger) cohorts. The plug-in efficient estimator has minimal bias and both $t$-based and FRT-based methods yield reliable inference for all specifications. The plug-in efficient estimator again dominates the other estimators in efficiency, although the gains are smaller (e.g. 20 to 30\% reductions in standard deviation relative to CS for complaints). The smaller efficiency gains in this specification are intuitive: the CS and SA estimators over-weight the pre-treatment periods (relative to the plug-in efficient estimator) in our setting, but the penalty for doing this is smaller in the collapsed data, where the pre-treatment outcomes are averaged over more months and thus have lower variance. 

\paragraph{Augmented $\hat{X}$.} Appendix Table \ref{tbl: long xhat simulation results} shows results for an alternative version of the plug-in efficient estimator where $\hat{X}$ is now a vector that contains the difference in means between cohort $g$ and $g'$ in all periods $t< min(g,g')$.\footnote{Calculation is more intensive using the longer $\hat{X}$, so we use 50 simulated permutations for the FRT, instead of the 500 used for the other specifications.} We find poor coverage of $t$-based CIs for this estimator in the monthly specification, where the dimension of $\hat{X}$ is large relative to the sample size (1975, compared with $N = 5537$), and thus the normal approximation derived in Proposition \ref{prop: asymptotic dist of thetahat betahat} is poor. By contrast, when the data is collapsed to the yearly level, and thus the dimension of $\hat X$ constructed in this way is more modest (only 10), the coverage for this estimator is good, and it offers small efficiency gains over the scalar $\hat X$ considered in the main text. These findings align with the results in \citet{lei_regression_2020}, who show (under certain regularity conditions) that covariate-adjustment in cross-sectional experiments yields asymptotically normal estimators when the dimensions of the covariates is $o(N^{\frac{1}{2}})$. We thus recommend using the version of $\hat{X}$ with all potential comparisons only when its dimension is small relative to the square root of the sample size.

\paragraph{Heterogeneous Treatment Effects.} Appendix Tables \ref{tbl: het effects - main} and \ref{tbl: het effects - ratio of sds} show simulation results for a modification of our baseline specification in which there are heterogeneous treatment effects. In the baseline specification, $Y_i(g) = Y_i(\infty)$ for all $g$. In the modification, we set $Y_i(g) = Y_i(\infty) + 1[t>=g] \cdot u_i$. The $u_i$ are mean-zero draws drawn from a normal distribution with standard deviation equal to the standard deviation of the untreated potential outcomes. We draw the $u_i$ once and hold them fixed throughout the simulations, which differ only in the assignment of treatment timing. The relative efficiency of the estimators is similar to those for the main specification, although as expected, both $t$-based and FRT-based approaches to inference tend to be conservative.

\begin{table}[!htb]
    \centering
    \captionsetup[table]{labelformat=empty,skip=1pt}
\begin{longtable}{llrrrrr}
\toprule
Estimator & Estimand & Bias & Coverage & FRT Size & Mean SE & SD \\ 
\midrule
PlugIn & calendar & 0.02 & 0.95 & 0.05 & 0.31 & 0.32 \\ 
PlugIn & cohort & 0.02 & 0.92 & 0.05 & 0.33 & 0.34 \\ 
PlugIn & ES0 & -0.02 & 0.95 & 0.05 & 0.34 & 0.34 \\ 
PlugIn & simple & 0.01 & 0.92 & 0.05 & 0.30 & 0.31 \\ 
CS & calendar & 0.01 & 0.95 & 0.05 & 0.55 & 0.55 \\ 
CS & cohort & 0.00 & 0.95 & 0.05 & 0.52 & 0.52 \\ 
CS/dCDH & ES0 & -0.01 & 0.96 & 0.05 & 0.46 & 0.46 \\ 
CS & simple & 0.01 & 0.95 & 0.05 & 0.52 & 0.52 \\ 
SA & calendar & 0.01 & 0.90 & 0.05 & 1.55 & 1.78 \\ 
SA & cohort & 0.00 & 0.88 & 0.07 & 1.63 & 1.86 \\ 
SA & ES0 & 0.01 & 0.94 & 0.06 & 0.97 & 1.03 \\ 
SA & simple & 0.02 & 0.87 & 0.06 & 1.77 & 2.04 \\ 
 \bottomrule
\end{longtable}

    \caption{Results for Simulations Calibrated to \citet{wood_reanalysis_2020} -- Use of Force\label{tbl: simulation results - main spec - force}}
    \floatfoot{Note: This table shows results analogous to Table \ref{tbl: simulation results - main spec}, except using Use of Force rather than Complaints as the outcome.}
\end{table} \setcounter{table}{\thetable -1}

\begin{table}[!htb]
    \centering
    \captionsetup[table]{labelformat=empty,skip=1pt}
\begin{longtable}{lrr}
\toprule
 & \multicolumn{2}{c}{Ratio of SD to Plug-In} \\ 
 \cmidrule(lr){2-3}
Estimand & CS & SA \\ 
\midrule
calendar & $1.71$ & $5.54$ \\ 
cohort & $1.55$ & $5.52$ \\ 
ES0 & $1.37$ & $3.05$ \\ 
simple & $1.69$ & $6.59$ \\ 
 \bottomrule
\end{longtable}

    \caption{Comparison of Standard Deviations -- \citet{callaway_difference--differences_2020} and \citet{sun_estimating_2020} versus Plug-in Efficient Estimator -- Use of Force}
    \label{tbl: comparison of SDs - main spec - force}
    \floatfoot{Note: This table shows results analogous to Table \ref{tbl: comparison of SDs - main spec}, except using Use of Force rather than Complaints as the outcome.}
\end{table} \setcounter{table}{\thetable -1}

\begin{table}[!htb]
    \centering
    \captionsetup[table]{labelformat=empty,skip=1pt}
\begin{longtable}{llrrrrr}
\toprule
Estimator & Estimand & Bias & Coverage & FRT Size & Mean SE & SD \\ 
\midrule
PlugIn & calendar & 0.00 & 0.95 & 0.06 & 0.05 & 0.06 \\ 
PlugIn & cohort & 0.00 & 0.94 & 0.04 & 0.04 & 0.04 \\ 
PlugIn & ES0 & 0.00 & 0.94 & 0.06 & 0.10 & 0.10 \\ 
PlugIn & simple & 0.00 & 0.94 & 0.04 & 0.04 & 0.04 \\ 
CS & calendar & 0.01 & 0.95 & 0.06 & 0.15 & 0.17 \\ 
CS & cohort & 0.01 & 0.96 & 0.05 & 0.14 & 0.14 \\ 
CS/dCDH & ES0 & 0.00 & 0.95 & 0.05 & 0.14 & 0.15 \\ 
CS & simple & 0.01 & 0.96 & 0.05 & 0.14 & 0.14 \\ 
SA & calendar & 0.02 & 0.77 & 0.05 & 0.40 & 0.48 \\ 
SA & cohort & 0.02 & 0.61 & 0.06 & 0.41 & 0.51 \\ 
SA & ES0 & 0.00 & 0.96 & 0.06 & 0.24 & 0.31 \\ 
SA & simple & 0.02 & 0.63 & 0.06 & 0.44 & 0.55 \\ 
 \bottomrule
\end{longtable}

    \caption{Results for Simulations Calibrated to \citet{wood_reanalysis_2020} -- Sustained Complaints}
    \label{tbl: simulation results - main spec - sustained}
    \floatfoot{Note: This table shows results analogous to Table \ref{tbl: simulation results - main spec}, except using Sustained Complaints rather than Complaints as the outcome.}
\end{table} \setcounter{table}{\thetable -1}

\begin{table}[!htb]
    \centering
    \captionsetup[table]{labelformat=empty,skip=1pt}
\begin{longtable}{lrr}
\toprule
 & \multicolumn{2}{c}{Ratio of SD to Plug-In} \\ 
 \cmidrule(lr){2-3}
Estimand & CS & SA \\ 
\midrule
calendar & $2.92$ & $8.38$ \\ 
cohort & $3.64$ & $13.83$ \\ 
ES0 & $1.46$ & $3.13$ \\ 
simple & $3.81$ & $14.68$ \\ 
 \bottomrule
\end{longtable}

    \caption{Comparison of Standard Deviations -- \citet{callaway_difference--differences_2020} and \citet{sun_estimating_2020} versus Plug-in Efficient Estimator -- Sustained Complaints}
    \label{tbl: comparison of SDs - main spec - sustained}
    \floatfoot{Note: This table shows results analogous to Table \ref{tbl: comparison of SDs - main spec}, except using Sustained Complaints rather than Complaints as the outcome.}
\end{table} \setcounter{table}{\thetable -1}


\begin{table}[!htb]
    \centering
    \captionsetup[table]{labelformat=empty,skip=1pt}
\begin{longtable}{llrrrrr}
\toprule
Estimator & Estimand & Bias & Coverage & FRT Size & Mean SE & SD \\ 
\midrule
PlugIn & calendar & 0.08 & 0.94 & 0.05 & 2.33 & 2.42 \\ 
PlugIn & cohort & 0.11 & 0.94 & 0.05 & 2.80 & 2.88 \\ 
PlugIn & ES0 & 0.12 & 0.93 & 0.07 & 2.30 & 2.41 \\ 
PlugIn & simple & 0.09 & 0.94 & 0.06 & 2.70 & 2.79 \\ 
CS & calendar & 0.02 & 0.96 & 0.04 & 3.20 & 3.15 \\ 
CS & cohort & 0.04 & 0.96 & 0.04 & 3.73 & 3.63 \\ 
CS/dCDH & ES0 & 0.08 & 0.95 & 0.05 & 2.89 & 2.89 \\ 
CS & simple & 0.03 & 0.96 & 0.04 & 3.68 & 3.61 \\ 
SA & calendar & -0.02 & 0.95 & 0.05 & 4.68 & 4.73 \\ 
SA & cohort & 0.00 & 0.96 & 0.04 & 5.04 & 4.93 \\ 
SA & ES0 & -0.03 & 0.95 & 0.05 & 4.38 & 4.39 \\ 
SA & simple & -0.01 & 0.96 & 0.04 & 5.20 & 5.14 \\ 
 \bottomrule
\end{longtable}

    \caption{Results for Simulations Calibrated to \citet{wood_reanalysis_2020} -- Annualized Data}
    \label{tbl: simulation results - annual - complaints}
    \floatfoot{Note: This table shows results analogous to Table \ref{tbl: simulation results - main spec}, except the data is collapsed to the annual level.}
\end{table} \setcounter{table}{\thetable -1}

\begin{table}[!htb]
    \centering
    \captionsetup[table]{labelformat=empty,skip=1pt}
\begin{longtable}{lrr}
\toprule
 & \multicolumn{2}{c}{Ratio of SD to Plug-In} \\ 
 \cmidrule(lr){2-3}
Estimand & CS & SA \\ 
\midrule
calendar & $1.30$ & $1.95$ \\ 
cohort & $1.26$ & $1.71$ \\ 
ES0 & $1.20$ & $1.82$ \\ 
simple & $1.29$ & $1.84$ \\ 
 \bottomrule
\end{longtable}

    \caption{Comparison of Standard Deviations -- \citet{callaway_difference--differences_2020} and \citet{sun_estimating_2020} versus Plug-in Efficient Estimator -- Annualized Data}
    \label{tbl: comparison of SDs - annual - complaints}
    \floatfoot{Note: This table shows results analogous to Table \ref{tbl: comparison of SDs - main spec}, except the data is collapsed to the annual level.}
\end{table} \setcounter{table}{\thetable -1}

\begin{table}[!htb]
    \centering
    \captionsetup[table]{labelformat=empty,skip=1pt}
\begin{longtable}{llrrrrr}
\toprule
Estimator & Estimand & Bias & Coverage & FRT Size & Mean SE & SD \\ 
\midrule
PlugIn & calendar & -0.16 & 0.95 & 0.05 & 2.71 & 2.69 \\ 
PlugIn & cohort & -0.17 & 0.94 & 0.05 & 3.23 & 3.24 \\ 
PlugIn & ES0 & -0.06 & 0.94 & 0.06 & 2.54 & 2.57 \\ 
PlugIn & simple & -0.18 & 0.94 & 0.05 & 3.15 & 3.16 \\ 
CS & calendar & -0.23 & 0.95 & 0.05 & 3.48 & 3.40 \\ 
CS & cohort & -0.29 & 0.95 & 0.05 & 4.06 & 3.99 \\ 
CS/dCDH & ES0 & -0.10 & 0.95 & 0.05 & 3.06 & 3.09 \\ 
CS & simple & -0.27 & 0.95 & 0.05 & 4.03 & 3.96 \\ 
SA & calendar & -0.16 & 0.94 & 0.06 & 5.04 & 5.01 \\ 
SA & cohort & -0.24 & 0.96 & 0.04 & 5.55 & 5.49 \\ 
SA & ES0 & -0.16 & 0.96 & 0.04 & 4.66 & 4.63 \\ 
SA & simple & -0.21 & 0.95 & 0.05 & 5.71 & 5.69 \\ 
 \bottomrule
\end{longtable}

    \caption{Results for Simulations Calibrated to \citet{wood_reanalysis_2020} -- Use of Force \& Annualized Data}
    \label{tbl: simulation results - annual - force}
    \floatfoot{Note: This table shows results analogous to Table \ref{tbl: simulation results - main spec}, except using Use of Force rather than Complaints as the outcome, and in simulations where data is collapsed to the annual level.}
\end{table} \setcounter{table}{\thetable -1}

\begin{table}[!htb]
    \centering
    \captionsetup[table]{labelformat=empty,skip=1pt}
\begin{longtable}{lrr}
\toprule
 & \multicolumn{2}{c}{Ratio of SD to Plug-In} \\ 
 \cmidrule(lr){2-3}
Estimand & CS & SA \\ 
\midrule
calendar & $1.26$ & $1.86$ \\ 
cohort & $1.23$ & $1.69$ \\ 
ES0 & $1.20$ & $1.80$ \\ 
simple & $1.25$ & $1.80$ \\ 
 \bottomrule
\end{longtable}

    \caption{Comparison of Standard Deviations -- \citet{callaway_difference--differences_2020} and \citet{sun_estimating_2020} versus Plug-in Efficient Estimator -- Use of Force \& Annualized Data}
    \label{tbl: comparison of SDs - annual - force}
    \floatfoot{Note: This table shows results analogous to Table \ref{tbl: comparison of SDs - main spec}, except using Use of Force rather than Complaints as the outcome, and in simulations where data is collapsed to the annual level.}
\end{table} \setcounter{table}{\thetable -1}

\begin{table}[!htb]
    \centering
    \captionsetup[table]{labelformat=empty,skip=1pt}
\begin{longtable}{llrrrrr}
\toprule
Estimator & Estimand & Bias & Coverage & FRT Size & Mean SE & SD \\ 
\midrule
PlugIn & calendar & 0.02 & 0.94 & 0.05 & 0.52 & 0.53 \\ 
PlugIn & cohort & 0.02 & 0.94 & 0.05 & 0.60 & 0.63 \\ 
PlugIn & ES0 & 0.04 & 0.94 & 0.06 & 0.67 & 0.67 \\ 
PlugIn & simple & 0.02 & 0.94 & 0.05 & 0.59 & 0.61 \\ 
CS & calendar & -0.02 & 0.95 & 0.06 & 0.86 & 0.88 \\ 
CS & cohort & -0.02 & 0.95 & 0.05 & 0.98 & 0.99 \\ 
CS/dCDH & ES0 & 0.01 & 0.96 & 0.05 & 0.88 & 0.85 \\ 
CS & simple & -0.02 & 0.95 & 0.06 & 0.97 & 0.99 \\ 
SA & calendar & 0.01 & 0.94 & 0.05 & 1.26 & 1.30 \\ 
SA & cohort & 0.01 & 0.95 & 0.05 & 1.34 & 1.37 \\ 
SA & ES0 & 0.02 & 0.95 & 0.05 & 1.31 & 1.33 \\ 
SA & simple & 0.01 & 0.95 & 0.05 & 1.39 & 1.43 \\ 
 \bottomrule
\end{longtable}

    \caption{Results for Simulations Calibrated to \citet{wood_reanalysis_2020} -- Sustained Complaints \& Annualized Data}
    \label{tbl: simulation results - annual - sustained}
    \floatfoot{Note: This table shows results analogous to Table \ref{tbl: simulation results - main spec}, except using Sustained Complaints rather than Complaints as the outcome, and in simulations where data is collapsed to the annual level.}
\end{table} \setcounter{table}{\thetable -1}

\begin{table}[!htb]
    \centering
    \captionsetup[table]{labelformat=empty,skip=1pt}
\begin{longtable}{lrr}
\toprule
 & \multicolumn{2}{c}{Ratio of SD to Plug-In} \\ 
 \cmidrule(lr){2-3}
Estimand & CS & SA \\ 
\midrule
calendar & $1.65$ & $2.45$ \\ 
cohort & $1.58$ & $2.18$ \\ 
ES0 & $1.27$ & $1.98$ \\ 
simple & $1.62$ & $2.34$ \\ 
 \bottomrule
\end{longtable}

    \caption{Comparison of Standard Deviations -- \citet{callaway_difference--differences_2020} and \citet{sun_estimating_2020} versus Plug-in Efficient Estimator -- Sustained Complaints \& Annualized Data}
    \label{tbl: comparison of SDs - annual - sustained}
    \floatfoot{Note: This table shows results analogous to Table \ref{tbl: comparison of SDs - main spec}, except using Sustained Complaints rather than Complaints as the outcome, and in simulations where data is collapsed to the annual level.}
\end{table} \setcounter{table}{\thetable -2}

\begin{table}
(a) Monthly Data
\captionsetup[table]{labelformat=empty,skip=1pt}
\begin{longtable}{llrrrrr}
\toprule
Estimator & Estimand & Bias & Coverage & FRT Size & Mean SE & SD \\ 
\midrule
PlugIn - Long X & calendar & -0.13 & 0.01 & 0.06 & 0.10 & 37.09 \\ 
PlugIn - Long X & cohort & 1.52 & 0.01 & 0.06 & 0.05 & 38.32 \\ 
PlugIn - Long X & ES0 & -6.19 & 0.04 & 0.05 & 0.26 & 119.39 \\ 
PlugIn - Long X & simple & 0.07 & 0.01 & 0.07 & 0.05 & 53.78 \\ 
PlugIn & calendar & 0.01 & 0.93 & 0.07 & 0.26 & 0.28 \\ 
PlugIn & cohort & 0.00 & 0.92 & 0.06 & 0.26 & 0.28 \\ 
PlugIn & ES0 & 0.00 & 0.96 & 0.04 & 0.32 & 0.31 \\ 
PlugIn & simple & 0.00 & 0.93 & 0.05 & 0.24 & 0.25 \\ 
 \bottomrule
\end{longtable}

(b) Annual Data
\captionsetup[table]{labelformat=empty,skip=1pt}
\begin{longtable}{llrrrrr}
\toprule
Estimator & Estimand & Bias & Coverage & FRT Size & Mean SE & SD \\ 
\midrule
PlugIn - Long X & calendar & 0.43 & 0.93 & 0.06 & 2.26 & 2.40 \\ 
PlugIn - Long X & cohort & 0.47 & 0.93 & 0.07 & 2.71 & 2.85 \\ 
PlugIn - Long X & ES0 & 0.49 & 0.92 & 0.07 & 2.23 & 2.38 \\ 
PlugIn - Long X & simple & 0.48 & 0.93 & 0.06 & 2.61 & 2.76 \\ 
PlugIn & calendar & 0.08 & 0.94 & 0.05 & 2.33 & 2.42 \\ 
PlugIn & cohort & 0.11 & 0.94 & 0.05 & 2.80 & 2.88 \\ 
PlugIn & ES0 & 0.12 & 0.93 & 0.07 & 2.30 & 2.41 \\ 
PlugIn & simple & 0.09 & 0.94 & 0.06 & 2.70 & 2.79 \\ 
 \bottomrule
\end{longtable}

\caption{Performance of Plug-In Efficient Estimator Using Augmented $\hat{X}$\label{tbl: long xhat simulation results}}
\floatfoot{Note: This table shows the bias, coverage, mean standard error, and standard deviation of two versions of the plug-efficient estimator. The estimator with the label ``Long X'' uses an augmented version of $\hat{X}$ that includes the difference in means between all cohorts $g,g'$ in periods $t<min(g,g').$ The estimator labeled PlugIn uses a scalar $\hat{X}$ such that the CS estimator corresponds with $\beta = 1$, as in the main text. The simulation specification in panel (a) is the baseline specification considered in the main text; in panel (b), the data is collapsed to the annual level.}
\end{table} \setcounter{table}{\thetable -1}

\begin{table}[!htb]
    \centering
    \captionsetup[table]{labelformat=empty,skip=1pt}
\begin{longtable}{llrrrrr}
\toprule
Estimator & Estimand & Bias & Coverage & FRT Size & Mean SE & SD \\ 
\midrule
PlugIn & calendar & 0.00 & 0.98 & 0.02 & 0.45 & 0.38 \\ 
PlugIn & cohort & 0.00 & 0.99 & 0.01 & 0.39 & 0.28 \\ 
PlugIn & ES0 & 0.00 & 0.99 & 0.01 & 0.42 & 0.31 \\ 
PlugIn & simple & 0.00 & 1.00 & 0.00 & 0.38 & 0.26 \\ 
CS & calendar & 0.00 & 0.97 & 0.03 & 0.63 & 0.58 \\ 
CS & cohort & 0.02 & 0.98 & 0.02 & 0.55 & 0.46 \\ 
CS/dCDH & ES0 & 0.00 & 0.99 & 0.01 & 0.52 & 0.43 \\ 
CS & simple & 0.02 & 0.98 & 0.02 & 0.55 & 0.47 \\ 
SA & calendar & -0.01 & 0.93 & 0.06 & 1.49 & 1.51 \\ 
SA & cohort & 0.01 & 0.91 & 0.07 & 1.54 & 1.58 \\ 
SA & ES0 & 0.00 & 0.97 & 0.03 & 0.96 & 0.94 \\ 
SA & simple & 0.01 & 0.90 & 0.07 & 1.67 & 1.72 \\ 
 \bottomrule
\end{longtable}

    \caption{Results for Simulations Calibrated to \citet{wood_reanalysis_2020} -- Heterogeneous Treatment Effects\label{tbl: het effects - main}}
    \floatfoot{Note: This table shows results analogous to Table \ref{tbl: simulation results - main spec}, except the DGP adds heterogeneous treatment effect as described in Section \ref{appendix section: additional sims}.}
\end{table} \setcounter{table}{\thetable -1}

\begin{table}[!htb]
    \centering
    \captionsetup[table]{labelformat=empty,skip=1pt}
\begin{longtable}{lrr}
\toprule
 & \multicolumn{2}{c}{Ratio of SD to Plug-In} \\ 
 \cmidrule(lr){2-3}
Estimand & CS & SA \\ 
\midrule
calendar & $1.71$ & $5.54$ \\ 
cohort & $1.55$ & $5.52$ \\ 
ES0 & $1.37$ & $3.05$ \\ 
simple & $1.69$ & $6.59$ \\ 
 \bottomrule
\end{longtable}

    \caption{Comparison of Standard Deviations -- \citet{callaway_difference--differences_2020} and \citet{sun_estimating_2020} versus Plug-in Efficient Estimator -- Heterogeneous Treatment Effects\label{tbl: het effects - ratio of sds}}
    \floatfoot{Note: This table shows results analogous to Table \ref{tbl: comparison of SDs - main spec}, except the DGP adds heterogeneous treatment effect as described in Section \ref{appendix section: additional sims}.}
\end{table} 

\clearpage 
\section{Additional Application Results\label{appendix: additional application results}}

This section contains additional results pertaining to our application in Section \ref{sec: wood application}. 

\subsection{Event-Study Results}

Appendix Figure \ref{fig: wood-et-al application event-study efficient} shows event-study estimates for the first two years after treatment using the plug-in efficient estimator, and Appendix Figure \ref{fig: wood-et-al application event-study CS} shows the analogous results using the CS estimator. Both plots show estimates for post-treatment effects as well as placebo estimates of pre-treatment effects (similar to pre-trends tests). In dark blue, we present point estimates and pointwise confidence intervals, and in light blue we present sup-$t$ simultaneous confidence bands \citep{olea_simultaneous_2019}.\footnote{We use the \texttt{suptCriticalValue} R package developed by Ryan Kessler.} It has been argued that simultaneous confidence bands are more appropriate for event-study analyses since they control size over the full dynamic path of treatment effects \citep{FreyaldenhovenEtAl(19),callaway_difference--differences_2020}. Both figures show that the simultaneous confidence bands include zero for nearly all periods for all three outcomes.

\begin{figure}[!hb]
    \centering
    \includegraphics[width = 0.9\linewidth]{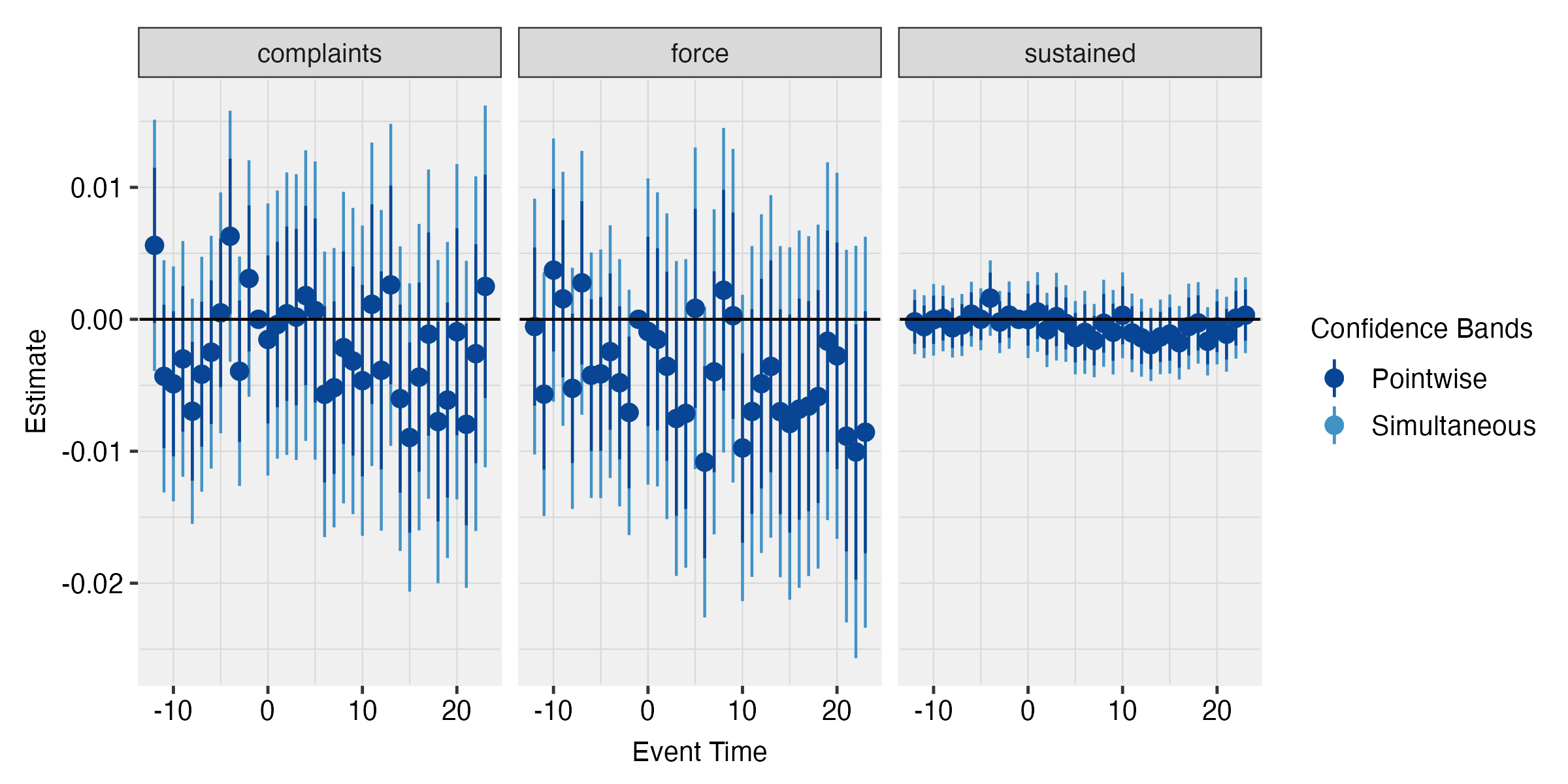}
    \caption{Event-Study Effects Using the Plug-In Efficient Estimator\label{fig: wood-et-al application event-study efficient}}
\end{figure}

\begin{figure}[!htb]
    \centering
    \includegraphics[width = 0.9\linewidth]{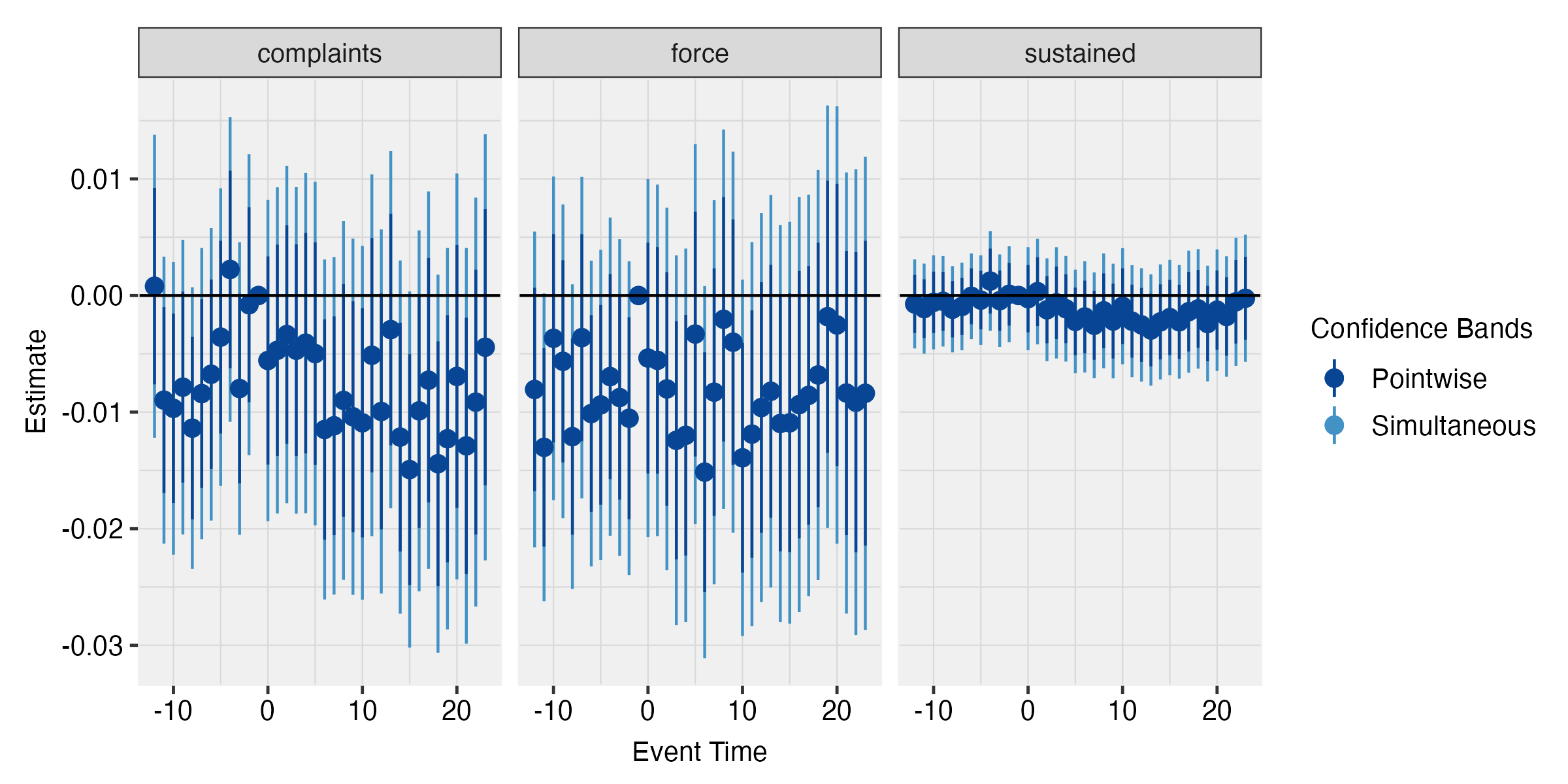}
    \caption{Event-Study Effects Using the CS Estimator}
    \label{fig: wood-et-al application event-study CS}
\end{figure}

\subsection{Balance and Robustness Checks}

Figure \ref{fig:balance age} shows a (binned) scatterplot of year of birth against training date for officers in our main analysis sample. The black circles show the raw data, and the orange triangles show the average over twenty equally-size bins. Overall, the average year of birth appears to be similar across all training dates. A univariate linear regression of year of birth on training month yields a coefficient of 0.01 --- implying that being trained a year later is associated with a 0.12 later birth year --- which is not statistically significant ($SE=$0.009; FRT $p$-value based on 5,000 permutations = 0.18). However, if we regress year of birth on dummies for training date, we obtain a $p$-value of 0.08 using an FRT for the hypothesis that all of the training dates are equal (FRT with 5,000 permutations). Thus, there may be some marginally significant imbalances in year of birth across training date, although they appear to be relatively small in magnitude and not systematically correlated with the timing of treatment. 

\begin{figure}[!ht]
    \centering
    \includegraphics[width = 0.9\textwidth]{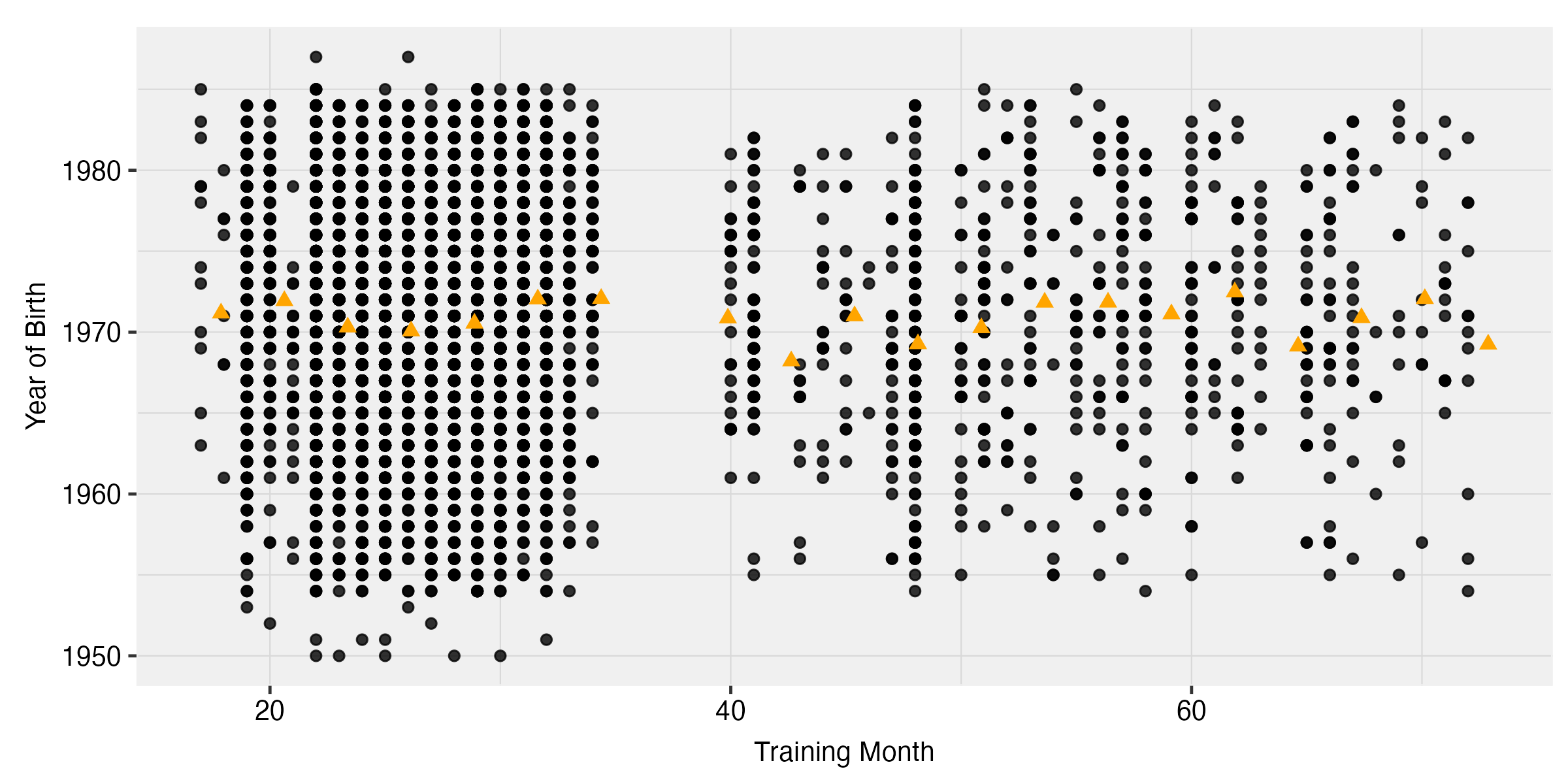}
    \caption{Covariate Balance on Age}
    \label{fig:balance age}
\end{figure}

In Appendix Figure \ref{fig: wood-et-al application summary comparison omit late g}, we present results analogous to those in Figure \ref{fig: wood-et-al application summary comparison} except removing officers who were treated in the last 12 months of the data. Appendix Table \ref{tbl:covariate balance - omit late treated} reports balance in $\hat{X}$ for this sample. The reason for focusing on this subsample is, as discussed in the supplement to \citet{wood_procedural_2020}, there was some non-compliance towards the end of the study period wherein officers who had not already been trained could volunteer to take the training at a particular date. As expected, we do not find any significant imbalance in $\hat{X}$ for this subsample. The qualitative patterns after dropping these observations are similar, although the estimates for the effect on use of force are not statistically significant in some specifications.

\begin{table}[!ht]
    \centering
    \caption{Tests of balance on pre-treatment outcomes - omitting late treated}
    \includegraphics[width = 0.9\textwidth]{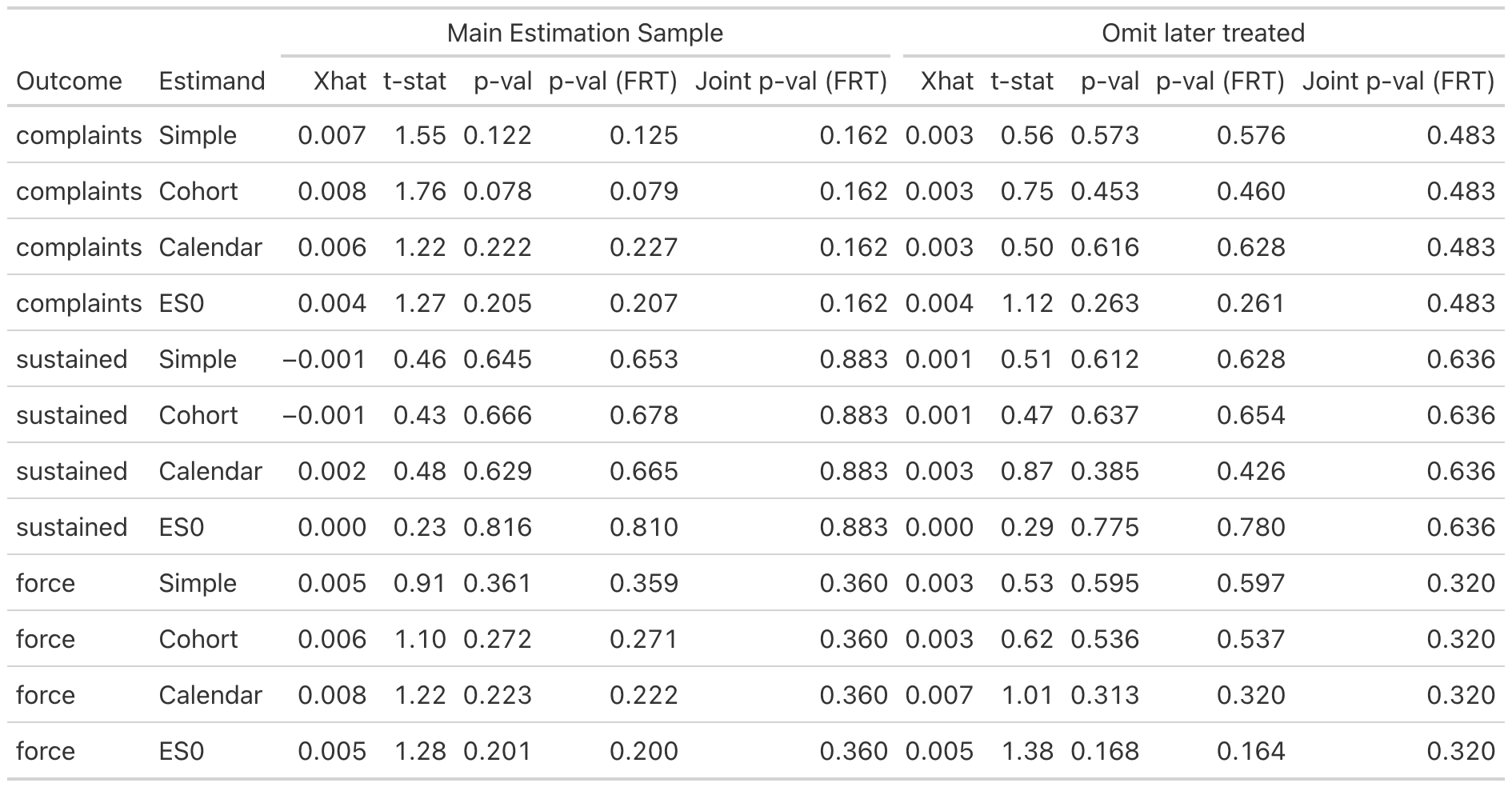}
    \label{tbl:covariate balance - omit late treated}
    \floatfoot{Note: this table is analogous to Table \ref{tbl:covariate balance}, except the rightmost columns show results omitting officers treated in the last year from the sample. All FRT $p$-values are based on 5,000 permutations.}
\end{table}

\begin{figure}[!hbtp]
    \centering
    \includegraphics[width=0.9\linewidth]{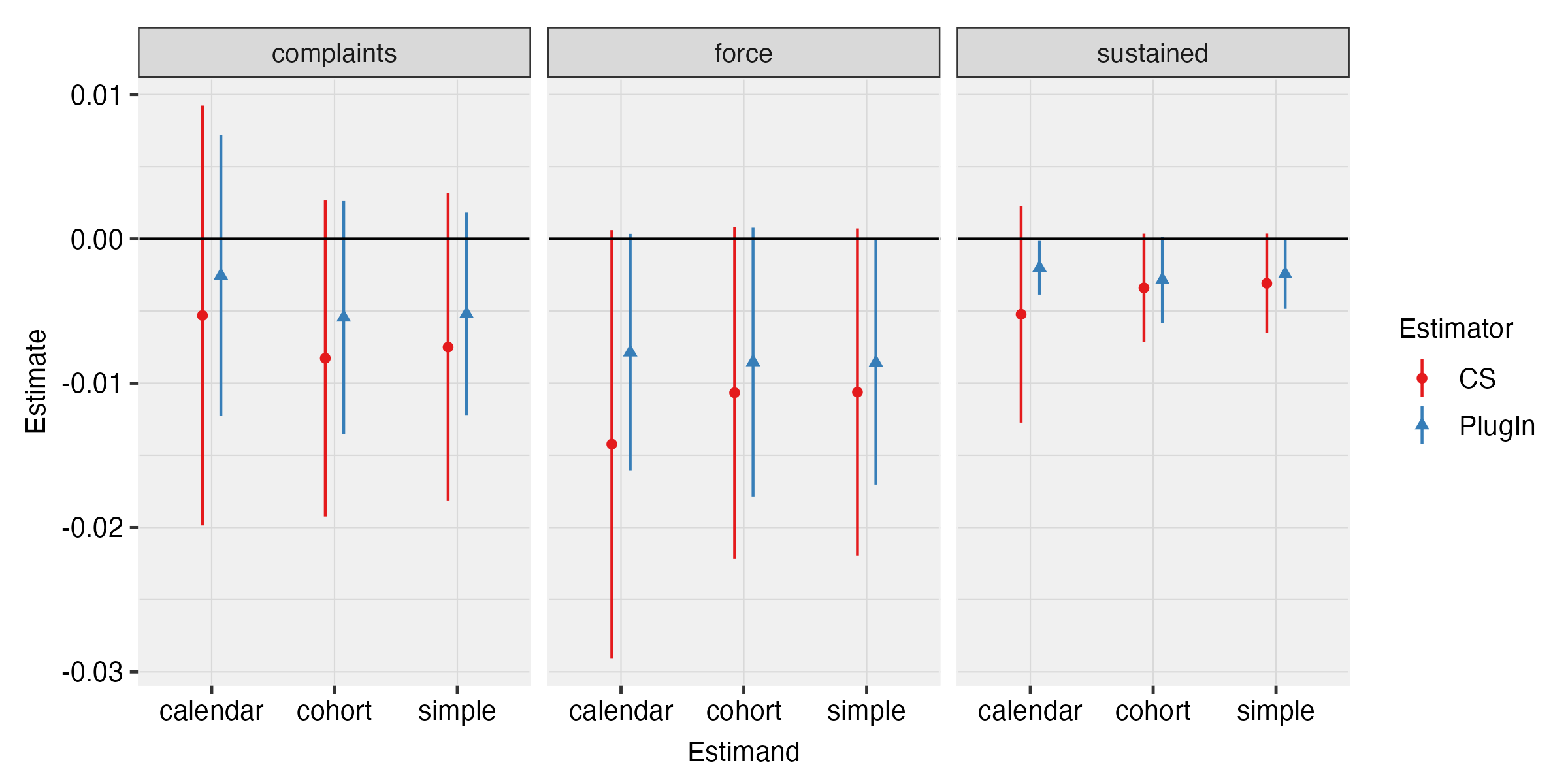}
    \caption{Effect of Procedural Justice Training Using the Plug-In Efficient and \citet{callaway_difference--differences_2020} Estimators -- Dropping Late-Trained Officers} \label{fig: wood-et-al application summary comparison omit late g}
    \floatfoot{Note: This figure is analogous to Figure \ref{fig: wood-et-al application summary comparison}, except we remove from the data officers trained in the last 12 months of the data owing to concerns about treatment non-compliance.}
\end{figure}
\clearpage 

\renewcommand{\refname}{Appendix References}
\putbib
\end{bibunit}
\end{document}